\documentclass[10pt,final,journal,twocolumn,twoside,letterpaper]{IEEEtran}

\usepackage[dvips]{graphicx}
\usepackage{cite}
\usepackage{amsmath}
\usepackage{times}
\usepackage{latexsym}
\usepackage{graphicx}
\usepackage{bm}
\usepackage{amssymb}
\usepackage[center]{caption2}
\usepackage{stfloats}
\usepackage{cases}
\usepackage{array}
\usepackage{setspace}
\usepackage{fancyhdr}
\usepackage{ulem}
\usepackage{flushend}
\usepackage{cite}
\usepackage{graphicx}
\usepackage{amsmath}
\usepackage{times}
\usepackage{latexsym}
\usepackage{graphicx}
\usepackage{bm}
\usepackage{amssymb}
\usepackage[center]{caption2}
\usepackage{stfloats}
\usepackage{cases}
\usepackage{array}
\usepackage{setspace}
\usepackage{fancyhdr}
\usepackage{algorithm}
\usepackage{algorithmic}

\usepackage{amsmath}
\usepackage{graphicx}
\usepackage{subfigure}

\usepackage{xcolor}

\usepackage{flushend}

\usepackage{fancyhdr}
\usepackage{layout}

\allowdisplaybreaks[4]

\newtheorem{theorem}{Theorem}

\newcounter{mytempeqncnt}

\newcaptionstyle{mycaptionstyle}{
  \captionlabel $.$ \: \captiontext \par}
\captionstyle{mycaptionstyle}

\title{Frequency Offset Estimation and \\ Training Sequence Design for MIMO OFDM
}

\author{{Yanxiang~Jiang,~\IEEEmembership{Student Member,~IEEE,}
Hlaing~Minn,~\IEEEmembership{Member,~IEEE,}
Xiqi~Gao,~\IEEEmembership{Senior Member,~IEEE,} Xiaohu~You,
and~Yinghui~Li~\IEEEmembership{Student Member,~IEEE}}
\thanks{\small {Manuscript received October 18, 2006; revised February 5, 2007;
accepted February 10, 2007. The associate editor coordinating the
review of this manuscript and approving it for publication was Dr.
Defeng (David) Huang. The work of Yanxiang~Jiang, Xiqi~Gao and
Xiaohu~You was supported in part by National Natural Science
Foundation of China under Grants 60496311 and 60572072, the China
High-Tech 863 Project under Grant 2003AA123310 and 2006AA01Z264, and
the International Cooperation Project on Beyond 3G Mobile of China
under Grant 2005DFA10360. The work of Hlaing~Minn and Yinghui~Li was
supported in part by the Erik Jonsson School Research Excellence
Initiative, the University of Texas at Dallas, USA. This paper was
presented in part at the IEEE International Conference on
Communications (ICC), Istanbul, Turkey, June 2006.}}
\thanks{\small {Yanxiang~Jiang, Xiqi~Gao and Xiaohu~You are with the National
Mobile Communications Research Laboratory, Southeast University,
Nanjing 210096, China (e-mail: \{yxjiang, xqgao, xhyu\}
@seu.edu.cn).}}
\thanks{\small {Hlaing~Minn and Yinghui~Li are with the Department of Electrical
Engineering, University of Texas at Dallas, TX 75083-0688, USA
(e-mail: \{hlaing.minn, yinghui.li\}@utdallas.edu).}}
}

\begin{document}
\maketitle

\begin{abstract}
    This paper addresses carrier frequency offset (CFO) estimation and
training sequence design for multiple-input multiple-output (MIMO)
orthogonal frequency division multiplexing (OFDM) systems over
frequency selective fading channels. By exploiting the orthogonality
of the training sequences in the frequency domain, integer CFO
(ICFO) is estimated. {With the uniformly spaced non-zero pilots in
the training sequences} and the corresponding geometric mapping,
fractional CFO (FCFO) is estimated through the roots of a real
polynomial. Furthermore, the condition for the training sequences to
guarantee estimation identifiability is developed. Through the
analysis of the correlation property of the training sequences, two
types of sub-optimal training sequences generated from the Chu
sequence are constructed. Simulation results verify the good
performance of the CFO estimator assisted by the proposed training
sequences.
\end{abstract}


\section{Introduction}

    Orthogonal frequency division multiplexing (OFDM)
transmission is receiving increasing attention in recent years due
to its robustness to frequency-selective fading and its
subcarrier-wise adaptability. On the other hand, multiple-input
multiple-output (MIMO) systems attract considerable interest due to
the higher capacity and spectral efficiency that they can provide in
comparison with single-input single-output (SISO) systems.
Accordingly, MIMO-OFDM has emerged as a strong candidate for beyond
third generation (B3G) mobile wide-band communications
\cite{Stuber}.

    It is well known that SISO-OFDM is highly sensitive to
carrier frequency offset (CFO), and accurate estimation and
compensation of CFO is very important \cite{Moose}. A number of
approaches have dealt with CFO estimation in a SISO-OFDM setup
\cite{Beek, Tureli, Moose, Morelli, David01, David02}. According to
whether the CFO estimators use training sequences or not, they can
be classified as blind ones \cite{Beek} \cite{Tureli} and
training-based ones \cite{Moose, Morelli, David01, David02}. Similar
to SISO-OFDM, MIMO-OFDM is also very sensitive to CFO. Moreover, for
MIMO-OFDM, there exists multi-antenna interference (MAI) between the
received signals from different transmit antennas. The MAI makes CFO
estimation more difficult, and a careful training sequence design is
required for training-based CFO estimation. However, unlike
SISO-OFDM, only a few works on CFO estimation for MIMO-OFDM have
appeared in the literature. In \cite{Ywyao}, a blind kurtosis-based
CFO estimator for MIMO-OFDM was developed. For training-based CFO
estimators, the overviews concerning the necessary changes to the
training sequences and the corresponding CFO estimators when
extending SISO-OFDM to MIMO-OFDM were provided in \cite{Mody,
Zelst}. However, with the provided training sequences in
\cite{Mody}, satisfactory CFO estimation performance cannot be
achieved. With the training sequences in \cite{Zelst}, the training
period grows linearly with the number of transmit antennas, which
results in an increased overhead. In \cite{Oliver}, a white sequence
based maximum likelihood (ML) CFO estimator was addressed for MIMO,
while a hopping pilot based CFO estimator was proposed for MIMO-OFDM
in \cite{Xiaoli}. Numerical calculations of the CFO estimators in
\cite{Oliver} \cite{Xiaoli} require a large point discrete Fourier
transform (DFT) operation and a time consuming line search over a
large set of frequency grids, which make the estimation
computationally prohibitive. To reduce complexity, computationally
efficient CFO estimation was introduced in \cite{Simoens} by
exploiting proper approximations. However, the CFO estimator in
\cite{Simoens} is only applied to flat-fading MIMO channels.


    When training sequence design for CFO estimation is concerned,
it has received relatively little attention. It was investigated for
single antenna systems in \cite{Stoica}, where a white sequence was
found to minimize the worst-case asymptotic Cramer-Rao bound (CRB).
Recently, an improved training sequence and structure design was
developed in \cite{Minn_Periodic} by exploiting the CRB and received
training signal statistics. In \cite{Ghogho}, training sequences
were designed for CFO estimation in MIMO systems using a
channel-independent CRB. In \cite{MinnFO}, the effect of CFO was
incorporated into the mean-square error (MSE) optimal training
sequence designs for MIMO-OFDM channel estimation in \cite{Minn1}.
Note that optimal training sequence design for MIMO-OFDM CFO
estimation in frequency selective fading channels is still an open
problem.

    In this paper, we propose a new training-based CFO estimator
for MIMO-OFDM over frequency selective fading channels. To avoid the
large point DFT operation and the time consuming line search and
also to guarantee the estimate precision, we propose to first
estimate integer CFO (ICFO) and then fractional CFO (FCFO). Next, we
relate FCFO estimation to direction of arrival (DOA) estimation in
uniform linear arrays (ULA). Then, we propose to exploit a geometric
mapping to transform the complex polynomial related to FCFO
estimation into a real polynomial. Correspondingly, FCFO is
estimated through the roots of the real polynomial. In the
derivation of the CFO estimator, we also develop the identifiability
condition concerning the training sequences. Furthermore, exploiting
the well-studied results for DOA estimation, we transform the
problem of training sequence design for CFO estimation into the
study of the correlation property of the training sequences, and
propose to construct them from the Chu sequence \cite{Chu}.

   The rest of this paper is organized as follows. In Section II,
we briefly describe the MIMO-OFDM system model. The proposed CFO
estimator including ICFO and FCFO estimation is presented in Section
III.  The aspect concerning training sequence design is addressed in
Section IV. Simulation results are shown in Section V. Final
conclusions are drawn in Section VI.

    \textsl{Notations}: Upper (lower) bold-face letters are used for matrices
(column vectors). Superscripts $* $, $T $ and $H $ denote conjugate,
transpose and Hermitian transpose, respectively. $( \cdot )_P $
denotes the residue of the number within the brackets modulo $P $.
$\Re(\cdot)$ and $\Im(\cdot)$ denote the real and imaginary parts of
the enclosed parameters, respectively. $\lfloor \cdot \rfloor$, $||
\cdot || ^2$, $\mathrm{E}[\cdot]$ and $\otimes$ denote the floor,
Euclidean norm-square, expectation and Kronecker product operators,
respectively. $\mathbf{sign}(\cdot)$ denotes the signum function and
$\mathbf{sign}(0) = 1$ is assumed. ${[\boldsymbol{x}]}_{m} $ denotes
the $m$-th entry of a column vector ${\boldsymbol{x}} $.
$\boldsymbol{x}^{(m)} $ denotes the $m$-cyclic-down-shift version of
$\boldsymbol{x}$ for $m>0$ and $|m|$-cyclic-up-shift version of
$\boldsymbol{x}$ for $m<0$. $\mathrm{diag}\{ \boldsymbol{x} \}$
denotes a diagonal matrix with the elements of $\boldsymbol{x}$ on
its diagonal. ${[\boldsymbol{X}]} _{m,n} $ denotes the $(m,n)$-th
entry of a matrix ${\boldsymbol{X}} $. $\boldsymbol{F}_N$ and
$\boldsymbol{I}_{N}$ denote the $N\times N$ {unitary DFT matrix} and
the $N\times N$ identity matrix, respectively.
$\boldsymbol{e}_N^{k}$ denotes the $k$-th column vector of
$\boldsymbol{I}_N$. $\boldsymbol{1}_Q$ ($\boldsymbol{0}_Q$) and
$\boldsymbol{0}_{P \times Q}$ denote the $Q \times 1$ all-one
(all-zero) vector and $P \times Q$ all-zero matrix, respectively.
$\boldsymbol{J}_Q$ denotes the $Q \times Q$ exchange matrix with
ones on its anti-diagonal and zeros elsewhere. {Unless otherwise
stated, $0 \le \mu \le N_t-1$ and $0 \le \nu \le N_r-1$ are
assumed.}

\section{System Model}

    Let us consider a MIMO-OFDM system with $N_t$ transmit antennas,
$N_r$ receive antennas and $N$ subcarriers. Suppose the training
sequence transmitted from the $\mu$-th antenna is denoted by the $N
\times 1$ vector $\tilde{\boldsymbol{t}} _{\mu} $. Before
transmission, this vector is processed by an inverse discrete
Fourier transform (IDFT), and a cyclic prefix (CP) of length $N_g$
is inserted. We assume that $N_g \ge L -1$, where $L$ is the maximum
length of all the frequency selective fading channels. We further
assume that all the transmit-receive antenna pairs are affected by
the same CFO. Define
\[
\boldsymbol{D} _{\bar{N}} (\varepsilon) = \mathrm{diag}\{[1,e^{j2\pi
\varepsilon /N} , \cdots ,e^{j2\pi \varepsilon (\bar{N}-1)/N}]^T \},
\]
where $\varepsilon$ is the frequency offset normalized by the
subcarrier spacing. Suppose the length-$L$ channel impulse response
from the $\mu$-th transmit antenna to the $\nu$-th receive antenna
is denoted by the $L \times 1$ vector $\boldsymbol{h}^{(\nu, \mu)}$.
Then, after removing the CP at the $\nu$-th receive antenna, the $N
\times 1$ received vector $\boldsymbol{y}_{\nu}$ can be written as
\cite{Xiaoli}
\begin{multline}
\boldsymbol{y}_{\nu} = \sqrt{N} e^{j{2\pi \varepsilon N_g}/{N}}
\boldsymbol{D}_{N}(\varepsilon) \sum\limits_{\mu  = 0}^{N_t  - 1}
\left\{\boldsymbol{F}_N^H \mathrm{diag}\{ \tilde
{\boldsymbol{h}}^{(\nu ,\mu )}\} \tilde {\boldsymbol{t}}_{\mu}
\right\}  \\ + \boldsymbol{w} _{\nu},
\end{multline}
where
\[
\tilde {\boldsymbol{h}}^{(\nu ,\mu )} = \boldsymbol{F}_N
[\boldsymbol{e}_N^{0 }, \boldsymbol{e}_N^{1 }, \cdots,
\boldsymbol{e}_N^{L-1}] \boldsymbol{h}^{(\nu ,\mu )},
\]
and $\boldsymbol{w}_{\nu}$ is an $N \times 1$ vector of additive
white complex Gaussian noise (AWGN) samples with zero-mean and equal
variance of $\sigma_w^2$.

    The goal of this paper is to design the training sequences $\{
\tilde{\boldsymbol{t}}_{\mu}  \} _{\mu = 0}^{N_t  - 1}$ and estimate
$\varepsilon$ from the observation of $\{ \boldsymbol{y}_{\nu} \}
_{\nu = 0}^{N_r - 1}$.

\section{Frequency Offset Estimation for MIMO OFDM}

    Let $Q=N/P$ with $(N)_{P} = 0$. Design
\begin{equation*}
(\mathrm{C}0) \ 0 \leq i_0 < i_1 < \cdots < i_{\mu} < \cdots <
i_{N_t-1} <Q.
\end{equation*}
Define
\[
\boldsymbol{\Theta} _{q} = [\boldsymbol{e}_{N} ^{q}, \boldsymbol{e}
_{N}^{q+Q}, \cdots,\boldsymbol{e}_{N}^ {q+(P-1)Q}], \ 0 \le q < Q.
\]
Let $\tilde{\boldsymbol{s}}_{\mu}$ denote a length-$P$ sequence
whose elements are all non-zero. Then, we propose to construct the
training sequence transmitted from the $\mu$-th antenna as
$\tilde{\boldsymbol{t}}_{\mu} = \boldsymbol{\Theta} _{i_{\mu}}
\tilde{\boldsymbol{s}}_{\mu}$ (C1). Without loss of generality, the
total energy allocated to training is supposed to be split equally
between the transmit antennas, i.e., $
||\tilde{\boldsymbol{s}}_{\mu}||^2 = N/N_t $ (C2). Note that the
entire training symbol is utilized by our proposed CFO estimator.

\begin{figure*}[!b]
\hrulefill
\normalsize
\setcounter{mytempeqncnt}{\value{equation}}
\setcounter{equation}{6}
\begin{multline}\label{5-2}
\bar{\boldsymbol{y}} = \sqrt{N} e^{j{2\pi \varepsilon_f N_g }/{N}}
\bigl \{\boldsymbol{I}_{N_r} \otimes \{[\boldsymbol{D}_N(\beta_0),
\boldsymbol{D}_N(\beta_1), \cdots, \boldsymbol{D}_N(\beta_{\mu}),
\cdots, \boldsymbol{D}_N(\beta_{N_t-1}) ]  \\
\times (\boldsymbol{I}_{N_t} \otimes \boldsymbol{1}_Q \otimes
\boldsymbol{F}_P^H) \mathrm{diag} \{ [\tilde{\boldsymbol{s}}_{0}^T,
\tilde{\boldsymbol{s}}_{1}^T, \cdots, \tilde{\boldsymbol{s}}
_{\mu}^T, \cdots, \tilde{\boldsymbol{s}} _{N_t-1}^T] ^T\}
\boldsymbol{\breve{F}} \} \bigl\} \boldsymbol{h} + \boldsymbol{v},
\end{multline}
\setcounter{equation}{\value{mytempeqncnt}}
\end{figure*}

\subsection{ICFO Estimation}

    Let $\boldsymbol{y} = [\boldsymbol{y}_{0}^T, \boldsymbol{y}_{1}^T,
\cdots, \boldsymbol{y}_{\nu}^T, \cdots, \boldsymbol{y}_{N_r-1}^T]^T$
denote the $N_r N \times 1$ cascaded vector from the $N_r$ receive
antennas. Then, $\boldsymbol{y}$ can be written as
\begin{equation}\label{3-01}
\boldsymbol{y}  = \sqrt{N} e^{j{{2\pi \varepsilon N_g }} / {N}} \{
\boldsymbol{I}_{N_r} \otimes [\boldsymbol{D}_N (\varepsilon )
\boldsymbol{S}]\} \boldsymbol{h} + \boldsymbol{w},
\end{equation}
where \setlength{\arraycolsep}{0.0em}
\begin{eqnarray*}\label{psiL}
\boldsymbol{{h}} &{}={}& [\boldsymbol{{h}}_{0}^T,
\boldsymbol{{h}}_{1}^T, \cdots, \boldsymbol{{h}}_{\nu}^T, \cdots,
\boldsymbol{{h}}_{N_r-1}^T]^T, \\
{\boldsymbol{h}}_{\nu} &{}={}& [({\boldsymbol{h}}^{(\nu,0)})^T,
({\boldsymbol{h}}^{(\nu,1)})^T, \cdots,
({\boldsymbol{h}}^{(\nu,\mu)})^T, \cdots,
({\boldsymbol{h}}^{(\nu,N_t-1)})^T ]^T, \\
\boldsymbol{S} &{}={}& {{\boldsymbol{\bar{F}}}} ^H \mathrm{diag} \{
[\tilde{\boldsymbol{s}}_{0}^T, \tilde{\boldsymbol{s}}_{1}^T, \cdots,
\tilde{\boldsymbol{s}} _{\mu}^T, \cdots, \tilde{\boldsymbol{s}}
_{N_t-1}^T] ^T\} {{\boldsymbol{\breve{F}}}}, \\
\boldsymbol{\bar{F}} &{}={}& [\boldsymbol{\Theta} _{i_0},
\boldsymbol{\Theta} _{i_1}, \cdots, \boldsymbol{\Theta} _{i_\mu},
\cdots, \boldsymbol{\Theta} _{i_{N_t -1} }] ^T \boldsymbol{{F}}_N , \\
{{\boldsymbol{\breve{F}}}} &{}={}& [ \boldsymbol{e}_{N_t}^{0}
\otimes \boldsymbol{\Theta}_{i_0}^T, \boldsymbol{e}_{N_t}^{1}
\otimes \boldsymbol{\Theta}_{i_1}^T, \cdots,
\boldsymbol{e}_{N_t}^{\mu} \otimes \boldsymbol{\Theta}_{i_\mu}^T,
\\
&{}{}& \cdots, \boldsymbol{e}_{N_t}^{N_t -1} \otimes
\boldsymbol{\Theta}_{i_{N_t -1}}^T]  \{ \boldsymbol{I}_{N_t} \otimes
[\boldsymbol{F}_N
[\boldsymbol{I}_L, \boldsymbol{0}_{L \times (N-L)}]^T] \} , \\
\boldsymbol{w} &{}={}& [\boldsymbol{w}_{0}^T, \boldsymbol{w}_{1}^T,
\cdots, \boldsymbol{w}_{\nu}^T, \cdots, \boldsymbol{w}_{N_r-1}^T]^T.
\end{eqnarray*}\setlength{\arraycolsep}{5pt}\hspace*{-3pt}

    From (\ref{3-01}), the ML estimation of $\varepsilon$ can be readily
obtained \cite{Kay} as follows
\begin{equation}\label{3-3}
\hat{\varepsilon} = \mathop {\arg \max
}\limits_{\tilde{\varepsilon}} \bigl\{ \boldsymbol{y}^H
\{\boldsymbol{I}_{N_r} \otimes [\boldsymbol{D}_N
(\tilde{\varepsilon} ) \boldsymbol{\tilde{S}}
(\boldsymbol{\tilde{S}}^H \boldsymbol{\tilde{S}})^{ - 1}
\boldsymbol{\tilde{S}}^H \boldsymbol{D}_N( - \tilde{\varepsilon}
)]\} \boldsymbol{y} \bigl\},
\end{equation}
where
\[
\boldsymbol{\tilde{S}} = {{\boldsymbol{\bar{F}}}}
 ^H \mathrm{diag} \{ [\tilde{\boldsymbol{s}}_{0}^T,
\tilde{\boldsymbol{s}}_{1}^T, \cdots, \tilde{\boldsymbol{s}}
_{\mu}^T, \cdots,\tilde{\boldsymbol{s}} _{N_t-1}^T] ^T\}.
\]
With condition (C0), which implies the orthogonality of the training
sequences in the frequency domain, matrix inversion involved in
(\ref{3-3}) can be avoided, and (\ref{3-3}) can thus be simplified
as follows
\begin{equation}\label{aequ1}
\hat{\varepsilon} = \mathop {\arg \max
}\limits_{\tilde{\varepsilon}} \bigl\{ || \{ \boldsymbol{I}_{N_r}
\otimes [ \boldsymbol{\bar{F}} \boldsymbol{D}_N ( -
\tilde{\varepsilon} )] \} \boldsymbol{y} ||^2 \bigl \}.
\end{equation}
Actually, with condition (C0), our proposed training sequences can
be treated as frequency-division multiplexing (FDM) pilot allocation
type sequences \cite{Minn1}.

    Let $\boldsymbol{l}$ denote the pilot location vector which is given by
$\boldsymbol{l}  =  \sum \limits_{\mu= 0} ^{N_t-1} {\boldsymbol {e}
_Q ^{i_{\mu}} } $. Then,  we have as follows.
\begin{theorem}\label{theorem1}
With $\tilde{\varepsilon}, \varepsilon \in ( -\lfloor Q/2\rfloor
,Q-\lfloor Q/2 \rfloor]$, $ \tilde{\varepsilon} = \varepsilon$
uniquely maximizes the cost function in (\ref{aequ1}) for
\textsl{any} $\boldsymbol{h} ^{(\nu,\mu)}(\ne \boldsymbol{0}_L)$
with the following condition
\begin{multline*}
\hspace*{-15pt} \mathrm{(C3)} \ (N-N_tP) \ge N_tP \ \& \ P \ge L \
\& \\
(\boldsymbol{1}_Q - \boldsymbol{l})^T \boldsymbol{l}^{(q)} >
0, \forall q \in \{1,2, \cdots, Q-1 \}.
\end{multline*}
\end{theorem}
\begin{proof}
See Appendix I.
\end{proof}
It follows immediately that the corresponding estimation is
identifiable for $\varepsilon \in ( -\lfloor Q/2\rfloor ,Q-\lfloor
Q/2 \rfloor]$ if condition (C3) is satisfied.

    From (\ref{aequ1}), the CFO can be estimated by exploiting
fast Fourier transform (FFT) interpolation and a time consuming line
search over a large set of frequency grids. But this approach is
computationally complicated and the estimate precision also depends
on the FFT size used. To reduce complexity, $\varepsilon $ is
divided into the ICFO $\varepsilon _i$ and the FCFO $\varepsilon
_f$. The ICFO $\varepsilon _i$ can be estimated by invoking only the
$N$ point FFT as follows
\begin{equation}\label{aequ2}
\hat{\varepsilon} _i  = \mathop {\arg \max }\limits_{{ -\lfloor
Q/2\rfloor < \tilde{\varepsilon}} _i \le Q-\lfloor Q/2 \rfloor}
\bigl \{ ||\{ \boldsymbol{I}_{N_r} \otimes [\boldsymbol{\bar{F}}
\boldsymbol{D}_N ( - {\tilde{\varepsilon}} _i )] \} \boldsymbol{y}
||^2 \bigl\}.
\end{equation}
The discussion on the uniqueness of $\hat{\varepsilon} _i =
\varepsilon _i $ in maximizing the estimation metric for
\textsl{any} $\boldsymbol{h} ^{(\nu,\mu)}(\ne \boldsymbol{0}_L)$ is
provided in Appendix II. Other ICFO estimators may also be used.
Once the ICFO is estimated, ICFO correction can be readily carried
out as follows
\begin{equation}\label{icfo6}
\bar{\boldsymbol{y}} =  { e^{-j{2\pi \hat{\varepsilon}_i N_g }/{N}}
  [\boldsymbol{I}_{N_r} \otimes
  \boldsymbol{D}_N(-\hat{\varepsilon}_i)]\boldsymbol{y}}.
\end{equation}

\begin{figure*}[!b]
\hrulefill
\normalsize
\setcounter{mytempeqncnt}{\value{equation}}
\setcounter{equation}{14}
\begin{equation}\label{Phiqq}
[\boldsymbol{\Phi}]_{q,q'} = j^{Q-1-q'} \sum \limits_{q'' =
\mathrm{max}\{0, q+q'-Q+1\}} ^{\mathrm{min}\{q, q'\}}
\left\{C_q^{q''} C_{Q-1-q}^{q'-q''} (-1)^{Q-1-q-q'+q''} \right\} , \
0 \le q, q' \le Q-1,
\end{equation}
\setcounter{equation}{\value{mytempeqncnt}}
\end{figure*}

\begin{figure*}[!b]
\vspace*{-10pt}
\normalsize
\setcounter{mytempeqncnt}{\value{equation}}
\setcounter{equation}{17} \setlength{\arraycolsep}{0.0em}
\begin{eqnarray}\label{psiL}
[\boldsymbol{\Phi}^H \boldsymbol{L}]_{q, q'} &{}={}& \sum \limits
_{q'' = 0 }^{[1-(Q)_2] Q/2 + (Q)_2 (Q-1)/2 -1} \left\{ [
\boldsymbol{\Phi}]_{q'',q}^* [\boldsymbol{L}]_{q'', q'} + [
\boldsymbol{\Phi}]_{q'',q} [\boldsymbol{L}]_{q'', q'}^* \right\}
+ (Q)_2 [\boldsymbol{\Phi}]_{(Q-1)/2,q} [\boldsymbol{L}]_{(Q-1)/2, q'} \nonumber \\
&{} = {}& 2 \Re \left\{ \sum \limits _{q'' = 0 }^{[1-(Q)_2] Q/2 +
(Q)_2 (Q-1)/2 -1} \left\{ [ \boldsymbol{\Phi}]_{q'',q}^*
[\boldsymbol{L}]_{q'', q'} \right\} \right\} + (Q)_2
[\boldsymbol{\Phi}]_{(Q-1)/2,q} [\boldsymbol{L}]_{(Q-1)/2, q'},
\nonumber \\
& & \hspace*{300pt}\forall q, q' \in \{ 0, 1, \cdots, Q-1 \}.
\end{eqnarray}\setlength{\arraycolsep}{5pt}\hspace*{-3pt}
\setcounter{equation}{\value{mytempeqncnt}}
\end{figure*}

\subsection{FCFO Estimation}
    The FCFO $\varepsilon _f$ is estimated based on
$\bar{\boldsymbol{y}}$. Assume $\hat{\varepsilon} _i  = \varepsilon
_i$ and substitute (\ref{3-01}) into (\ref{icfo6}). Then,  by
exploiting condition (C0) and condition (C1), which implies the
periodic property of the training sequences, $\bar{\boldsymbol{y}}$
can be expressed in an equivalent form as shown in (\ref{5-2}) at
the bottom of the page. In (\ref{5-2}),
\setlength{\arraycolsep}{0.0em}
\begin{eqnarray*}
\beta_\mu &{}={}& \varepsilon_f + \: i_\mu,\\
\boldsymbol{v} &{}={}& e^{ - j{2\pi \varepsilon _i N_g }/{N}}
[\boldsymbol{I}_{N_r} \otimes \boldsymbol{D}_N ( - \varepsilon _i )
] \boldsymbol{w}.
\end{eqnarray*}
\setlength{\arraycolsep}{5pt}\hspace*{-3pt} It follows from
(\ref{5-2}) that the estimation of $\varepsilon_f$ is equivalent to
the estimation of the $N_t$ different equivalent CFOs $\{\beta_\mu\}
_{\mu=0}^{N_t-1}$.

\addtocounter{equation}{1}

    Furthermore, exploiting the periodic property of the training
sequences again, we can stack $\bar{\boldsymbol{y}}$ into the $
Q\times N_r P$ matrix $\boldsymbol{Y} = [\boldsymbol{Y}_{0},
\boldsymbol{Y}_{1}, \cdots, \boldsymbol{Y}_{\nu}, \cdots
\boldsymbol{Y}_{N_r-1}]$, where
\[
[\boldsymbol{Y}_{\nu}]_{q,p} = [((\boldsymbol{e} _{N_r} ^{\nu}) ^T
\otimes \boldsymbol{I}_N) \bar{\boldsymbol{y}}]_{qP+p},\ 0 \le q <
Q, 0 \le p < P.
\]
Then, (\ref{5-2}) can be expressed in the following equivalent form
\begin{equation}\label{5-3}
\boldsymbol{Y} = \boldsymbol{B}\boldsymbol{X} + \boldsymbol{V},
\end{equation}
where \setlength{\arraycolsep}{0.0em}
\begin{eqnarray*}
\boldsymbol{B} &{}={}& [\boldsymbol{b}_0, \boldsymbol{b}_1, \cdots,
\boldsymbol{b}_{\mu}, \cdots, \boldsymbol{b}_{N_t -1}], \\
\boldsymbol{b}_{\mu} &{}={}& [1, e ^{j {2\pi \beta_\mu}/{Q} },
\cdots, e ^{j {2\pi \beta_\mu q}/{Q} }, \cdots, e ^{j {2\pi
\beta_\mu(Q-1)}/ {Q} }]^T, \\
\boldsymbol{X} &{}={}& [\boldsymbol{X}_{0}, \boldsymbol{X}_{1},
\cdots, \boldsymbol{X}_{\nu}, \cdots, \boldsymbol{X}_{N_r -1}], \\
\boldsymbol{X}_{\nu} &{}={}& [\boldsymbol{x}^{(\nu,0)},
\boldsymbol{x}^{(\nu,1)}, \cdots, \boldsymbol{x}^{(\nu,\mu)},
\cdots, \boldsymbol{x}^{(\nu, N_t-1)}]^T, \\
\boldsymbol{x}^{(\nu, \mu)} &{}={}&  \sqrt{P} e^{j{2\pi
\varepsilon_f N_g }/{N}} \boldsymbol{D}_{P} (\beta _\mu ) \\
&{} {}& \times \boldsymbol{F}_P^H \mathrm{diag}
\{\tilde{\boldsymbol{s}}_{\mu}\} \boldsymbol{\Theta} _{i_\mu}^T
\boldsymbol{F}_N [\boldsymbol{I}_L, \boldsymbol{0}_{L \times
(N-L)}]^T \boldsymbol{h} ^{(\nu,\mu)},
\end{eqnarray*}\setlength{\arraycolsep}{5pt}\hspace*{-3pt}
and $\boldsymbol{V}$ is
the $ Q\times N_r P$ matrix generated from $\boldsymbol{v}$ in the
same way as $\boldsymbol{Y}$. From (\ref{5-3}), we can see that
there exists an inherent relationship between FCFO estimation and
DOA estimation in ULA \cite{Petre} with $P$, $Q$ and $N_r$
satisfying $N_rP > Q$ (C4).

    The covariance matrix of $\boldsymbol{Y}$ can
be estimated by $\boldsymbol{\hat{R}} _{\boldsymbol{Y}
\boldsymbol{Y}} = \boldsymbol{Y} \boldsymbol{Y}^H / (N_r P)$. Let
$\boldsymbol{L}$ denote the $Q \times Q$ unitary column conjugate
symmetric matrix which is given by
\begin{multline*} \boldsymbol{L}
= \frac{[1- (Q)_2]}{{\sqrt 2 }}\left[ {\begin{array}{*{20}c}
   {\boldsymbol{I}_{Q/2} } & {j\boldsymbol{I}_{Q/2} }  \\
   {\boldsymbol{J}_{Q/2} } & { - j\boldsymbol{J}_{Q/2} }  \\
\end{array}} \right] \\
+ \frac{(Q)_2}{{\sqrt 2 }}\left[ {\begin{array}{*{20}c}
   {\boldsymbol{I}_{(Q-1)/2} } & { \boldsymbol{0}_{(Q-1)/2} } & {j\boldsymbol{I}_{(Q-1)/2} }  \\
   { \boldsymbol{0}_{(Q-1)/2}^T } & { \sqrt{2} } & { \boldsymbol{0}_{(Q-1)/2}^T} \\
   {\boldsymbol{J}_{(Q-1)/2} } & { { \boldsymbol{0}_{(Q-1)/2} } } & { - j\boldsymbol{J}_{(Q-1)/2} }  \\
\end{array}} \right].
\end{multline*}
Then,  with the aid of $\boldsymbol{L}$, the complex matrix
$\boldsymbol{\hat{R}} _{\boldsymbol{Y} \boldsymbol{Y}}$ can be
transformed into a real matrix \cite{Marius} as follows
\setlength{\arraycolsep}{0.0em}
\begin{eqnarray}
{\boldsymbol{\hat{R}}} _{\boldsymbol{Y} \boldsymbol{Y}}^r &{}={}&
{1}/{2} \cdot \boldsymbol{L}^H (\boldsymbol{\hat{R}}_{\boldsymbol{Y}
\boldsymbol{Y}} +\: \boldsymbol{J}_Q
\boldsymbol{\hat{R}}_{\boldsymbol{Y} \boldsymbol{Y}}^*
\boldsymbol{J}_Q)\boldsymbol{L} \nonumber \\
&{}={}& \Re ( \boldsymbol{L}^H \boldsymbol{\hat{R}} _{\boldsymbol{Y}
\boldsymbol{Y}} \boldsymbol{L}) .
\end{eqnarray}
\setlength{\arraycolsep}{5pt}\hspace*{-3pt} The eigen-decomposition
of the real matrix ${\boldsymbol{\hat{R}}} _{\boldsymbol{Y}
\boldsymbol{Y}}^r$ can be obtained as
\begin{equation}
 {\boldsymbol{\hat{R}}} _{\boldsymbol{Y} \boldsymbol{Y}}^r
 = \boldsymbol{E}_{\boldsymbol{X}} \boldsymbol{\Lambda} _{\boldsymbol{X}}
 \boldsymbol{E}_{\boldsymbol{X}}^H
 + \boldsymbol{E}_{\boldsymbol{V}} \boldsymbol{\Lambda} _{\boldsymbol{V}}
 \boldsymbol{E}_{\boldsymbol{V}}^H ,
\end{equation}
where \setlength{\arraycolsep}{0.0em}
\begin{eqnarray*}
\boldsymbol{\Lambda} _{\boldsymbol{X}}&{} ={} &
\mathrm{diag}\{[\lambda_0, \lambda_1, \cdots,
\lambda_{N_t-1}]^T\},\\
\boldsymbol{\Lambda} _{\boldsymbol{V}} &{}={}& \sigma _w^2
\boldsymbol{I} _{Q-N_t}, \\
\lambda_0 &{}\geq{}& \lambda_1\geq\cdots\geq \lambda_{N_t-1} >
\sigma _w^2,
\end{eqnarray*}
\setlength{\arraycolsep}{5pt}\hspace*{-3pt} and
$\boldsymbol{E}_{\boldsymbol{X}}$ and
$\boldsymbol{E}_{\boldsymbol{V}}$ contain the unitary eigen-vectors
that span the signal space and noise space, respectively. Let
\setlength{\arraycolsep}{0.0em}
\begin{eqnarray*}
z &{} ={} & e^{j 2 \pi \beta / Q},\\
\boldsymbol{a}(z ) &{}={}& [ 1, z, \cdots, z^{Q-1}]^T.
\end{eqnarray*}
\setlength{\arraycolsep}{5pt}\hspace*{-3pt} Then,  by exploiting
$\boldsymbol{a}(z )$, a polynomial of degree $2(Q-1)$ with complex
coefficients can be obtained \cite{YxjiangICC06} as follows
\begin{equation}\label{5-8}
f(z) = \boldsymbol{a}^T (z) \boldsymbol{J}_Q \boldsymbol{A}
\boldsymbol{a} (z),
\end{equation}
where
\[
\boldsymbol{A} =\boldsymbol{L} [ \boldsymbol{I}_Q -
\boldsymbol{E}_{\boldsymbol{X}} \boldsymbol{E}_{\boldsymbol{X}}^H
]\boldsymbol{L}^H.
\]
Correspondingly, $ \{ \beta _\mu  \} _{\mu  = 0}^{N_t  - 1} $ can be
indirectly estimated by calculating the pairwise roots of $f(z)=0$
which are closest to the unit circle.

    Note that by its definition, $z$ is always located on the unit
circle. By analysis, we find that the following geometric mapping
exists,
\begin{equation}
g(z) = \cot (\pi \beta/{Q} ) = j(z + 1)/(z - 1).
\end{equation}
It follows immediately that $g(z)$ is a monotonic reversible mapping
with real values for $\beta \in [-0.5,Q-0.5)$. Hence, we have $z(g)
= (g+j)/(g-j)$. Then, $\boldsymbol{a} (z)$ can be expressed as a
function of $g$ as
\begin{equation}\label{bz}
\boldsymbol{a} (z) = (g - j)^{1 - Q}\boldsymbol{d} (g),
\end{equation}
where
\[
\boldsymbol{d} (g) = [(g-j)^{Q-1}, (g+j)(g-j)^{Q-2},\cdots,
(g+j)^{Q-1}]^T.
\]
Since the elements of $\boldsymbol{d} (g)$ are polynomials of degree
$(Q-1)$ with respect to $g$, $\boldsymbol{d} (g)$ in (\ref{bz}) can
be further expressed as
\begin{equation}\label{dqg}
\boldsymbol{d} (g) = {\boldsymbol{\Phi}} \boldsymbol{a} (g),
\end{equation}
where ${\boldsymbol{\Phi}}$ is the $Q \times Q$ coefficient matrix.
By straight-forward calculation, the elements of $\boldsymbol{\Phi}$
are obtained as shown in (\ref{Phiqq}). \addtocounter{equation}{1}
In (\ref{Phiqq}), $C_q^{q''} = q!/[(q-q'')!{q''}!] $. Then, from
(\ref{Phiqq}), we immediately obtain
\begin{multline}
[\boldsymbol{\Phi}]_{q,q'} = [\boldsymbol{\Phi}]_{Q-1-q,q'}^* \ \& \
(Q)_2 \Im \bigl\{ [\boldsymbol{\Phi}]_{(Q-1)/2, q} \bigl\} = 0, \\
\forall q, q' \in \{ 0, 1, \cdots, Q-1 \},
\end{multline}
which shows that ${\boldsymbol{\Phi}}$ is a column conjugate
symmetric matrix \textsl{no matter} $Q$ is even or odd.

    By exploiting (\ref{bz}) and (\ref{dqg}) in (\ref{5-8}), the following
equivalent polynomial can be obtained (ignoring the constant items)
\begin{equation}\label{newfz}
{f}^r(g) = \boldsymbol{a} ^T (g) \boldsymbol{J}_Q \boldsymbol{A}^r
\boldsymbol{a} (g)  ,
\end{equation}
where
\[
\boldsymbol{A}^r = \boldsymbol{J}_Q \boldsymbol{\Phi}^H
\boldsymbol{L} [ \boldsymbol{I}_Q - \boldsymbol{E}_{\boldsymbol{X}}
\boldsymbol{E}_{\boldsymbol{X}}^H ]\boldsymbol{L}^H
{\boldsymbol{\Phi}}.
\]
Due to the column conjugate symmetric property of $\boldsymbol{L}$
and ${\boldsymbol{\Phi}}$, we have (\ref{psiL}) as shown at the
bottom of the page. \addtocounter{equation}{1} Hence,
$\boldsymbol{\Phi}^H \boldsymbol{L}$ becomes a real matrix
\textsl{no matter} $Q$ is even or odd. Moreover,
$\boldsymbol{E}_{\boldsymbol{X}}$ is also a real matrix.
Accordingly, ${f}^r(g)$ is transformed into a polynomial with real
coefficients. The roots of ${f}^r(g) = 0$ can be obtained by the
fast root-calculating algorithms for real polynomials in
\cite{Press}, whose computational complexities are much less than
those for complex ones.

    After the roots of ${f}^r(g) = 0$ are obtained, the FCFO
can be readily estimated according to the following steps.
\begin{enumerate}
\item Find the $N_t$ pairwise roots of ${f}^r(g) = 0$ whose imaginary
parts have the smallest absolute values, $ \{ \Re(g_{\mu}) \pm j \Im
(g_{\mu}) \}_{\mu=0}^{N_t-1}$.

\item Calculate the equivalent CFOs corresponding to the $N_t$
roots,
\begin{equation}
\hat{\beta} _{\mu}  = \left({Q}/{\pi } \times
\mathrm{acot}[\Re(g_{\mu}) ]\right) _Q.
\end{equation}

\item Calculate the FCFO $\hat{\varepsilon} _{\mu}^f$ corresponding to
each $\hat{\beta}_{\mu}$,
\begin{equation}
\hat{\varepsilon} _{\mu}^f = \left\{ {\begin{array}{*{3}l}
   {\hat{\beta} _{\mu}  - i_0,} & {\hat{\beta} _{\mu} \in ((i_0 - \epsilon _{\mathrm{th}})_Q ,  i_0 + \epsilon _{\mathrm{th}})}  \\
   {\hat{\beta} _{\mu}  - i_1 ,} & { \hat{\beta} _{\mu} \in (i_1 - \epsilon _{\mathrm{th}} , i_1 + \epsilon _{\mathrm{th}})}  \\
    \cdots  &  \cdots   \\
   {\hat{\beta} _{\mu}  - i_{N_t-1},} & {\hat{\beta} _{\mu} \in (i_{N_t-1} - \epsilon
   _{\mathrm{th}}, } \\
   {} & {\hspace*{40pt} (i_{N_t-1} + \epsilon _{\mathrm{th}})_Q)}  \\
\end{array}} \right.,
\end{equation}
where $\epsilon _{\mathrm{th}}$ denotes a threshold which is
predefined to avoid the ambiguous estimation and $0.5 < \epsilon
_{\mathrm{th}} < 1$, $[a,b)$ has the usual meaning for $a<b$, and
$[a,b) = [0,b) \cup [a, Q)$ for $a>b $.

\item Obtain the final FCFO by averaging $\{\hat{\varepsilon}
_{\mu}^f \} _{{\mu}=0}^{N_t-1}$,
\begin{equation}
\hat{\varepsilon} _f  = \frac{1}{N_t} \sum\limits_{{\mu} = 0}^{N_t -
1} {\hat{\varepsilon} _{\mu}^f }.
\end{equation}

\end{enumerate}

\subsection{Computational Complexity}

    The computational load of our proposed CFO estimator mainly
involves the $N$ point FFT, the eigen-decomposition of
${\boldsymbol{\hat{R}}} _{\boldsymbol{Y} \boldsymbol{Y}}^r$ and the
root-calculation for ${f}^r(g) = 0$, which require $4N \log_2 N$, $9
Q^3$ and ${64}/{3} \cdot (Q-1)^3$ real additions or multiplications
\cite{Press} \cite{Golub}, respectively. Compared with the direct
CFO estimate from (\ref{aequ1}), which requires a large point DFT
operation and an exhaustive line search, our proposed estimator has
significantly lower complexity especially with a relatively small
$Q$. Furthermore, by applying the approach in \cite{FeifeiGao}, the
complexity of our proposed estimator can be further decreased by
calculating the roots from the first-order derivative of ${f}^r(g) =
0$, but at the cost of a slight performance degradation at low
signal-to-noise ratio (SNR).

\section{Training Sequence Design}

    With our design conditions in the previous section, the training
sequences are determined completely by the base sequences $\{
\boldsymbol{\tilde{s}}_{\mu} \}_{\mu =0} ^{N_t -1}$ with fixed $P$,
$Q$ and $\{i_{\mu} \} _{\mu =0} ^{N_t-1}$. Let
$\boldsymbol{R}_{\boldsymbol{X} \boldsymbol{X}}$ denote the
covariance matrix of $\boldsymbol{X}$. It is pointed out in
\cite{Petre} \cite{Marius} that the optimal performance can be
achieved with uncorrelated signals, i.e., diagonal matrix
$\boldsymbol{R} _{\boldsymbol{X} \boldsymbol{X}}$. Since
$\boldsymbol{R} _{\boldsymbol{X} \boldsymbol{X}}$ can be estimated
by $\boldsymbol{\hat{R}} _{\boldsymbol{X} \boldsymbol{X}} =
\boldsymbol{X} \boldsymbol{X}^H / (N_r P)$, it is expected that good
performance can be achieved with diagonal matrix
$\boldsymbol{\hat{R}} _{\boldsymbol{X} \boldsymbol{X}}$. Making
$\boldsymbol{\hat{R}} _{\boldsymbol{X} \boldsymbol{X}}$ to be a
diagonal matrix is equivalent to making $[\boldsymbol{\hat{R}}
_{\boldsymbol{X} \boldsymbol{X}}]_{\mu, \mu '} \bigl | _{\mu \neq
\mu'} = 0 $. Define \setlength{\arraycolsep}{0.0em}
\begin{eqnarray*}
\varpi_{\mu, \mu '} &{} ={} & ( i_{\mu }-i_{\mu
'})/Q, \\
\boldsymbol{s}_{\mu} &{}={}& \frac{1}{\sqrt{Q}} \boldsymbol{F}_P^H
\boldsymbol{\tilde{s}}_{\mu}.
\end{eqnarray*}
\setlength{\arraycolsep}{5pt}\hspace*{-5pt} Assume that the channel
taps remain constant during the training period. Then,
$[\boldsymbol{\hat{R}} _{\boldsymbol{X} \boldsymbol{X}}]_{\mu, \mu
'}$ can be expressed as
\begin{multline}\label{equ30}
[\boldsymbol{\hat{R}} _{\boldsymbol{X} \boldsymbol{X}}]_{\mu, \mu '}
=  \\
\frac{Q}{N_r} \sum\limits_{\nu =0}^{N_r -1}\sum\limits_{l
=0}^{L-1} \sum\limits_{l' =0}^{L-1} \left\{[\boldsymbol{h}^{(\nu,
\mu)}]_l [\boldsymbol{T}^{(\mu, \mu ')}]_{l, l'}
[\boldsymbol{h}^{(\nu, \mu ')}]_{l'}^* \right\},
\end{multline}
where
\[
[\boldsymbol{T}^{(\mu, \mu ')}]_{l, l'} = (\boldsymbol{s}_{\mu}
^{(l)})^T \boldsymbol{D}_P(\varpi_{\mu, \mu '} Q)
(\boldsymbol{s}_{\mu'} ^{(l')})^*.
\]
Since the elements of $\boldsymbol{h} ^{(\nu, \mu)}$ are random
variables, $[\boldsymbol{\hat{R}} _{\boldsymbol{X}
\boldsymbol{X}}]_{\mu, \mu '}\bigl | _{\mu \neq \mu'} = 0 $ can be
achieved by making the training sequences satisfy
\begin{equation}\label{th2cond}
[\boldsymbol{T}^{(\mu, \mu ')}]_{l, l'} = 0, \ \mathrm{if} \ \mu \neq
\mu ' \ \& \  0 \le l, l'  \le L-1.
\end{equation}

    It follows immediately that the optimal training
sequences should satisfy the condition in (\ref{th2cond}). However,
with condition (C0), $\varpi_{\mu, \mu '}$ is definitely a decimal
fraction, which greatly complicates the satisfaction of the
condition in (\ref{th2cond}) with proper training sequences. To ease
the above problem, we construct two types of sub-optimal training
sequences which can make the off-diagonal elements of
$\boldsymbol{\hat{R}} _{\boldsymbol{X} \boldsymbol{X}}$ as small as
possible. It has been proven in \cite{Minn_Periodic} that training
sequence with the zero auto-correlation (ZAC) property is optimal
for CFO estimation in SISO frequency selective fading channels.
Besides, constant amplitude ZAC (CAZAC) sequence (e.g., \cite{Chu,
WHMow_SSTA} and references therein) is often a preferred choice for
training. Therefore, we propose to construct the training sequences
from a length-$P$ Chu sequence $\boldsymbol{s}$ with its element
given by
\[
[\boldsymbol{s}]_p = e^{j\pi v p^2 / P},\ 0 \le p \le P-1,
\]
where $v$ is coprime to $P$, and $P$ is supposed to be even.

    Let $\boldsymbol{\tilde{s}}_{\mu } = \sqrt{Q/N_t} \boldsymbol{F}_P
\boldsymbol{s}^{(\mu M)}$ with $M = \lfloor P/N_I \rfloor$ and $N_I
\ge N_t$. Then,  we refer to the so-constructed training sequences
as TS 0. Let $\boldsymbol{\tilde{s}}_{\mu } = \sqrt{Q/N_t}
\boldsymbol{F}_P \boldsymbol{s}$. Then,  we refer to the
so-constructed training sequences as TS 1. Note that the TS 1
training sequences are equivalent to the so-called repeated
phase-rotated Chu (RPC) sequences \cite{Coon}. Define
\[
p_{\mu, l} = (1-m)\mu M + l,
\]
where $m=0$ for TS 0, and $m=1$ for TS 1. Then,  from the above
constructions we have \setlength{\arraycolsep}{0.0em}
\begin{eqnarray}\label{Tuull}
[\boldsymbol{T}^{(\mu, \mu ')}]_{l, l'} &{} ={} & \frac{1}{N_t}
(-1)^{v(p_{\mu , l} - p_{\mu ', l'}) +1} e^{j \pi v (p_{\mu, l} ^2 -
p_{\mu ', l'}^2) / P} \nonumber\\
& & \times e^{-j \pi (P-1)[v(p_{\mu , l} - p_{\mu ', l'}) -
\varpi_{\mu, \mu '}]/P} \sin ( \pi \varpi_{\mu, \mu '} ) \nonumber \\
& &  / \sin \{ \pi [v(p_{\mu, l} - p_{\mu ', l'}) - \varpi_{\mu, \mu
'}]/P \}.
\end{eqnarray}
\setlength{\arraycolsep}{5pt}\hspace*{-3pt} It follows immediately
from (\ref{Tuull}) that
\begin{multline}\label{equ36}
\bigl |[\boldsymbol{T}^{(\mu, \mu ')}]_{l, l'}\bigl | _{p_{\mu, l}-
p_{\mu ', l'} \ne 0} \\
\ll  \bigl |[\boldsymbol{T}^{(\mu, \mu
')}]_{l, l'}\bigl |_{p_{\mu, l}- p_{\mu ', l'} = 0} < {P}/{N_t}.
\end{multline}
We see that $\bigl |[\boldsymbol{T}^{(\mu, \mu ')}]_{l, l'}\bigl |$
achieves its maximum for TS 0 when $(\mu-\mu') M = l' -l$ and for TS
1 when $l = l'$. Assume that the channel energy is mainly
concentrated in the preceding $M$ channel taps and the first channel
tap is the dominant one. Then, it can be inferred from (\ref{equ30})
and (\ref{equ36}) that the value of $\bigl |[\boldsymbol{\hat{R}}
_{\boldsymbol{X} \boldsymbol{X}}]_{\mu, \mu '}\bigl |_{\mu \ne \mu
'}$ for TS 0 is very small, and it is much smaller than that for TS
1 in the same channel environment. In this sense, we can say that TS
0 is superior to TS 1, which will be verified through simulation
results in the following section.

    Note that our proposed training sequence structure is similar
to the one introduced recently in \cite{Ghogho}, and the
identifiability of our CFO estimator, however, cannot be guaranteed
with the training sequences in \cite{Ghogho}.

\section{Simulation Results}

    To evaluate our CFO estimator's performance with the proposed
training sequences for MIMO-OFDM, a number of simulations are
carried out. Throughout the simulations, a MIMO OFDM system of
bandwidth $20$MHz operating at 5GHz with $N=1024$ and $N_g = 64$ is
used. Each channel has $4$ independent Rayleigh fading taps, whose
relative average-powers and propagation delays are $\{0, -9.7,
-19.2, -22.8\}$dB and $\{0, 0.1, 0.2, 0.4\}\mu $s, respectively. For
the training sequences, we set $P$ and $Q$ to $64$ and $16$,
{respectively.} The normalized CFO $\varepsilon$ is generated within
the range $(-\lfloor Q/2 \rfloor, Q- \lfloor Q/2 \rfloor]$. For
description convenience, we henceforth refer to the proposed real
polynomial based CFO estimator as RPBE.

    In the following, we use the average CRB (avCRB), which corresponds
to the extended Miller and Chang bound (EMCB) \cite{Minn_Periodic}
\cite{Gini}, to benchmark the performance of RPBE. The average CRB
or EMCB is obtained by simply averaging the snapshot CRB over
independent channel realizations, which can be calculated
straight-forwardly \cite{Kay} as follows
\begin{multline}
\mathrm{CRB}_\varepsilon  =  \\
\frac{N\sigma_w^2} {8 \pi^2 \boldsymbol{h}^H
\boldsymbol{\mathcal{X}}^H \boldsymbol{\mathcal{B}}
[\boldsymbol{I}_{N_rN} - \boldsymbol{\mathcal{X}}
(\boldsymbol{\mathcal{X}}^H \boldsymbol{\mathcal{X}})^{-1}
\boldsymbol{\mathcal{X}}^H] \boldsymbol{\mathcal{B}}
\boldsymbol{\mathcal{X}} \boldsymbol{h}},
\end{multline}
where \setlength{\arraycolsep}{0.0em}
\begin{eqnarray*}
\boldsymbol{\mathcal{X}}  & {} ={} & \boldsymbol{I}_{N_r} \otimes \boldsymbol{S}, \\
\boldsymbol{\mathcal{B}} &{} = {} & \boldsymbol{I}_{N_r} \otimes
\mathrm{diag}\{[N_g, N_g+1, \cdots, N_g +N -1]^T \}.
\end{eqnarray*}\setlength{\arraycolsep}{5pt}\hspace*{-3pt}

\begin{figure}[!b]
\centering
\includegraphics[width=0.48\textwidth]{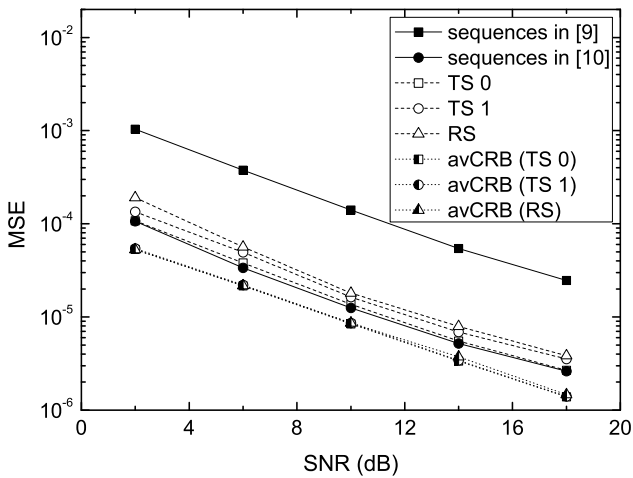}
\captionstyle{mycaptionstyle} \caption{CFO estimation performance
for different training sequences with $N_t = 3$ and $N_r = 2$.}
\label{Jiang_Paper-TW-Oct-06-0846_fig1}
\end{figure}

    In order to verify our analysis concerning training
sequence design with the Chu sequence, we construct the random
sequences (RS) by generating the $N_tP$ pilots randomly, and compare
the MSE performances of RPBE with TS 0, TS 1 and RS. Also included
for comparison are the performances of the CFO estimators with their
training sequences in \cite{Mody} and \cite{Zelst}. In Fig.
\ref{Jiang_Paper-TW-Oct-06-0846_fig1}, we present the corresponding
simulation results with $N_t =3$ and $N_r =2$. It can be observed
that the performance of RPBE with TS 0 is slightly better than that
with TS 1, which coincides with the analytical results concerning
training sequence design in the previous section. It can also be
observed that the performances of RPBE with TS 0 and TS 1 are better
than that with RS, which should be attributed to the good
correlation property of TS 0 and TS 1. Actually, the performance
improvements are more evident when the frequency selective fading
channels with large delay spreads are considered. Furthermore, we
observe that the performances of RPBE with TS 0 and TS 1 are far
better than that of the CFO estimator in \cite{Mody}, and almost the
same as that of the CFO estimator in \cite{Zelst}. Noting that the
overhead of the training sequences presented in \cite{Zelst} grows
linearly with the number of transmit antennas, we can find certain
advantages of the proposed training sequences. We also observe that
the average CRB for TS 0 is slightly smaller than that for TS 1.
Since the average CRBs for TS 0 and TS 1 are very close, only one
curve is plotted subsequently.

\begin{figure}[!b]
\centering
\includegraphics[width=0.48\textwidth]{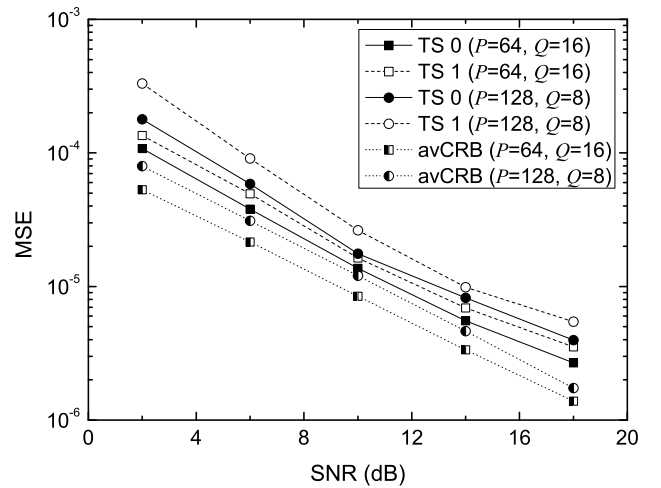}
\captionstyle{mycaptionstyle} \caption{{CFO estimation performance
for different $Q$ with $N_t = 3$ and $N_r = 2$.}}
\label{Jiang_Paper-TW-Oct-06-0846_fig2}
\end{figure}

    Depicted in Fig. \ref{Jiang_Paper-TW-Oct-06-0846_fig2} is the MSE performance of RPBE
as a function of SNR for different $Q$ {with $N_t =3$ and $N_r =2$.}
Studying the curves in Fig. \ref{Jiang_Paper-TW-Oct-06-0846_fig2},
we can see that the performance of RPBE with the same training
sequences degrades with smaller $Q$ provided that $Q > N_t $ and $P
\ge L $, which agrees well with the analysis in \cite{Petre}. We can
also see that a smaller $Q$ yields a larger average CRB, and the MSE
performance follows the same trend as the average CRB.

   In Fig. \ref{Jiang_Paper-TW-Oct-06-0846_fig3}, we illustrate how the number of
transmit antennas $N_t$ affects the performance of RPBE. It can be
observed that the MSE performance for TS 1 deteriorates with
increased $N_t$ due to the influence of MAI, whereas the MSE
performance for TS 0 degrades slightly, which should be ascribed to
the better correlation property of TS 0. Another observation is that
the number of transmit antennas has little impact on the average
CRB.

    Fig. \ref{Jiang_Paper-TW-Oct-06-0846_fig4} shows how the number of receive antennas
$N_r$ affects the performance of RPBE. We can observe that the MSE
performance of RPBE is substantially improved for both TS 0 and TS 1
by increasing $N_r$. We have the similar observation for the average
CRB. These observations imply that it is more efficient for
performance improvement of RPBE to increase the number of receive
antennas than the number of transmit antennas.

\begin{figure}[!t]
\centering
\includegraphics[width=0.48\textwidth]{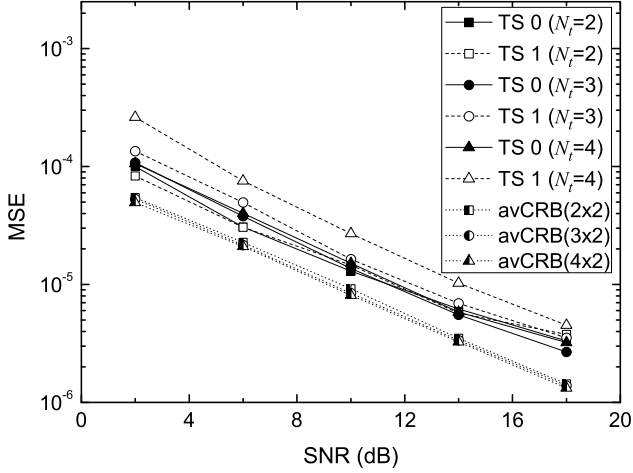}
\captionstyle{mycaptionstyle} \caption{{CFO estimation performance
for different numbers of transmit antennas with $N_r =2$.}}
\label{Jiang_Paper-TW-Oct-06-0846_fig3}
\end{figure}

\begin{figure}[!t]
\centering
\includegraphics[width=0.48\textwidth]{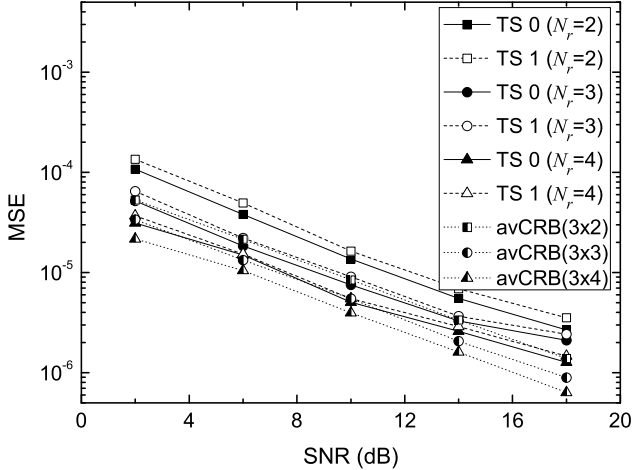}
\captionstyle{mycaptionstyle} \caption{{CFO estimation performance
for different numbers of receive antennas with $N_t =3$.}}
\label{Jiang_Paper-TW-Oct-06-0846_fig4}
\end{figure}

\section{Conclusions}

    In this paper, we have presented a training sequence assisted
CFO estimator for MIMO OFDM systems by exploiting the properties of
the training sequences and an efficient geometric mapping. We have
developed the required conditions for the training sequences to
yield estimation identifiability, and proposed sub-optimal training
sequences constructed from the Chu sequence. Our proposed estimator
and training sequences yield better or similar estimation
performance with much smaller training overhead than existing
methods for MIMO-OFDM systems.

\begin{figure*}[!b]
\hrulefill
\normalsize
\setcounter{mytempeqncnt}{\value{equation}}
\setcounter{equation}{33}
\begin{multline}\label{equ34}
\prod \limits_{\mu=1}^{N_t-1} \prod \limits_{\mu'=0}^{\mu-1} \bigl
\{( b_{\mu} - b_{\mu'} )( a_{\mu'} - a_{\mu} )\bigl \} = \prod
\limits_{\mu=2}^{N_t-1} \prod \limits_{\mu'=1}^{\mu-1} ( b_{\mu} -
b_{\mu'} ) \\
\times \sum \limits_{\mu '' = 0}^{N_t -1} \biggl \{ (-1)^{\mu ''}
\cdot \prod \limits_{\mu=1, \mu \ne \mu ''}^{N_t-1} \prod
\limits_{\mu'=0, \mu' \ne \mu ''}^{\mu-1} ( a_{\mu'} - a_{\mu} )
\cdot \prod \limits_{\mu=1}^{N_t-1} ( a_{\mu ''} - b_{\mu} ) \cdot
\prod \limits_{\mu=0, \mu \ne \mu ''}^{N_t-1} ( a_{\mu} - b_{0} )
\biggl \}.
\end{multline}
\setcounter{equation}{\value{mytempeqncnt}}
\end{figure*}

\appendices
\section{}

    This appendix presents the proof of Theorem 1. Substitute
(\ref{3-01}) into the cost function in (\ref{aequ1}) and ignore the
noise item. Define $\Delta = \varepsilon -\tilde {\varepsilon} $.
Then, the cost function in (\ref{aequ1}) can be expressed as
\begin{multline}
\mathfrak{P}(\Delta) = \\
N \boldsymbol{h} ^H \{\boldsymbol{I}_{N_r}
\otimes [ \boldsymbol{S}^H \boldsymbol{D}_N (-\Delta)
\boldsymbol{\bar{F}}^H \boldsymbol{\bar{F}} \boldsymbol{D}_N
(\Delta) \boldsymbol{S} ] \} \boldsymbol{h}.
\end{multline}
Let
\[
\boldsymbol{\breve{S}} = \mathrm{diag} \{
[\tilde{\boldsymbol{s}}_{0}^T, \tilde{\boldsymbol{s}}_{1}^T, \cdots,
\tilde{\boldsymbol{s}} _{\mu}^T, \cdots, \tilde{\boldsymbol{s}}
_{N_t-1}^T] ^T\}.
\]
Define
\[
\mathfrak{D}(\Delta) = N \boldsymbol{h}^H [\boldsymbol{I}_{N_r}
\otimes ( \boldsymbol{\breve{F}}^H \boldsymbol{\breve{S}}^H
\boldsymbol{\breve{S}} \boldsymbol{\breve{F}} )] \boldsymbol{h}
-\mathfrak{P}(\Delta).
\]
Then, it follows immediately that the maximum of
$\mathfrak{P}(\Delta)$ corresponds to the minimum of
$\mathfrak{D}(\Delta)$ with respect to $\Delta$. Suppose
\begin{multline*}
0 \le \omega < Q-N_t \ \& \ 0 \le z_{\omega} < Q \ \& \ z_{\omega}
\ne i_{\mu}  \ \& \\
z_{\omega} = z_{\omega '}, \ \mathrm{iff.} \
\omega = \omega '.
\end{multline*}
Then,  $\mathfrak{D}(\Delta)$ can be expressed as
\begin{equation}
\mathfrak{D}(\Delta) = \boldsymbol{d}^H (\Delta)
\boldsymbol{d}(\Delta),
\end{equation}
where \setlength{\arraycolsep}{0.0em}
\begin{eqnarray*}
\boldsymbol{d}(\Delta)  & {} ={} & [\boldsymbol{d}_{0}^T(\Delta),
\boldsymbol{d}_{1}^T(\Delta), \cdots, \boldsymbol{d}
_{\nu}^T(\Delta), \cdots, \boldsymbol{d} _{N_r-1}^T (\Delta)] ^T, \\
\boldsymbol{d} _{\nu}(\Delta) &{} = {} & \boldsymbol{\mathcal{C}}
\boldsymbol{\breve{S}} \boldsymbol{\breve{F}} \boldsymbol{{h}}
_{\nu}, \\
\boldsymbol{\mathcal{C}} & {} ={} & \left[ {\begin{array}{*{20}l}
   {\boldsymbol{\mathcal{C}}^{(0,0)} } &  \cdots  & {\boldsymbol{\mathcal{C}}^{(0,N_t  - 1)} }  \\
    \vdots  &  \boldsymbol{\mathcal{C}}^{(\omega, \mu)}  &  \vdots   \\
   {\boldsymbol{\mathcal{C}}^{(Q - N_t  - 1,0)} } &  \cdots  & {\boldsymbol{\mathcal{C}}^{(Q - N_t  - 1,N_t  - 1)} }  \\
 \end{array} } \right], \\
\boldsymbol{\mathcal{C}}^{(\omega, \mu)} & {} ={} & {\sqrt{N}}
\boldsymbol{\Theta}_{z_\omega}^T \boldsymbol{F} _N \boldsymbol{D}_N
(\Delta) \boldsymbol{F} _N^H \boldsymbol{\Theta}_{i_\mu}.
\end{eqnarray*}\setlength{\arraycolsep}{5pt}\hspace*{-3pt}
Correspondingly, we have
\begin{multline}
\mathfrak{D}(\Delta) \geq 0 \ \& \ \mathfrak{D}(0) = 0 \ \& \\
\mathfrak{D}(\Delta) = 0, \ \mathrm{iff.} \ \boldsymbol{d}(\Delta) =
\boldsymbol{0}_{N_r(N-N_tP)}.
\end{multline}
Therefore, for \textsl{any} $\boldsymbol{h} ^{(\nu,\mu)}(\ne
\boldsymbol{0}_L)$, proving that $\tilde \varepsilon = \varepsilon$
is the unique value to maximize $\mathfrak{P}(\Delta)$ with $\tilde
\varepsilon, \varepsilon \in ( -\lfloor Q/2\rfloor ,Q-\lfloor Q/2
\rfloor]$ is equivalent to proving that $\boldsymbol{d}(\Delta) \ne
\boldsymbol{0}_{N_r(N-N_tP)}$ for \textsl{any} $\Delta \in (-Q, 0)
\cup (0, Q)$.

    Consider the following two cases:

    1) When $\Delta$ is not an integer:  We establish immediately
from the definition of $\boldsymbol{\mathcal{C}} ^{(\omega, \mu)}$
that it is a column-wise circulant matrix. Denote the square matrix
formed by the first $N_t P$ rows of $\boldsymbol{\mathcal{C}}$ by
$\boldsymbol{\mathcal{C}} _{\mathcal{P}}$ (recall that $(N-N_tP) \ge
N_tP$ in our design condition). Then, by exploiting the relation
between circulant matrix and DFT matrix \cite{Golub},
$\boldsymbol{\mathcal{C}} _{\mathcal{P}}$ can be decomposed as
\begin{equation}\label{app35}
\boldsymbol{\mathcal{C}} _{\mathcal{P}} = \sqrt{P} (\boldsymbol{I}
_{N_t} \otimes \boldsymbol{F} _{P}^H) \boldsymbol{\Lambda}
(\boldsymbol{I} _{N_t} \otimes \boldsymbol{F} _{P}),
\end{equation}
where \setlength{\arraycolsep}{0.0em}
\begin{eqnarray*}
\boldsymbol{\Lambda} &{} = {} & \left[ {\begin{array}{*{20}l}
{\boldsymbol{\Lambda}^{(0,0)} } &  \cdots  & {\boldsymbol{\Lambda}^{(0,N_t  - 1)} }  \\
\vdots  &  \boldsymbol{\Lambda}^{(\mu, \mu')}  &  \vdots   \\
{\boldsymbol{\Lambda}^{(N_t-1,0)} } & \cdots & {\boldsymbol{\Lambda}^{(N_t-1,N_t-1)} }  \\
\end{array} } \right], \\
\boldsymbol{\Lambda}^{(\mu, \mu')} &{} = {}& \mathrm{diag}\{
\boldsymbol{\tilde{c}}^{(\mu, \mu')} \}, \\
\boldsymbol{\tilde{c}}^{(\mu, \mu')} & {} = {} & \boldsymbol{F} _{P}
\boldsymbol{\mathcal{C}}^{(\mu, \mu')} \boldsymbol{e} _{P}^0.
\end{eqnarray*}\setlength{\arraycolsep}{5pt}\hspace*{-3pt}
The element of $\boldsymbol{\tilde{c}}^{(\mu, \mu')}$ can be
obtained from its definition as follows
\setlength{\arraycolsep}{0.0em}
\begin{eqnarray}\label{lambda}
[\boldsymbol{\tilde{c}}^{(\mu, \mu')}]_p &{} = {} &
\frac{1}{\sqrt{Q}} c^{(\mu, \mu')} {(1- e^{j 2 \pi \Delta})} e^{j 2
\pi (i_{\mu'} -
z_{\mu} + \Delta) (-p)_P /N} \nonumber\\
&{} \ne {}& 0, \hspace*{90pt} 0 \le p \le \ P-1,
\end{eqnarray} \setlength{\arraycolsep}{5pt}\hspace*{-3pt}
where \setlength{\arraycolsep}{0.0em}
\begin{eqnarray*}\label{lambda}
c^{(\mu, \mu')} &{} = {} & {a_{\mu}}/({a_{\mu} - b_{\mu'}}), \\
a_{\mu} &{} = {}& e^{j2 \pi z_{\mu} /Q}, \\
b_{\mu'} &{} = {}& e^{j2 \pi (i_{\mu'} +\Delta) /Q}.
\end{eqnarray*} \setlength{\arraycolsep}{5pt}\hspace*{-3pt}
It follows from (\ref{lambda}) that $ \mathrm{rank}
\{\boldsymbol{\Lambda}^{(\mu, \mu')} \} = P$. Let
\[
\boldsymbol{\Lambda}_{p} = [\boldsymbol{I}_{N_t} \otimes (
\boldsymbol{e}_P^p )^T ] \boldsymbol{\Lambda} [\boldsymbol{I}_{N_t}
\otimes ( \boldsymbol{e}_P^p )^T ]^T,
\]
which denotes the $N_t \times N_t$ sub-matrix of
$\boldsymbol{\Lambda} $. Then, we can decompose
$\boldsymbol{\Lambda}_{p}$ as follows
\begin{equation}\label{theor26}
\boldsymbol{\Lambda}_{p} = \frac{1}{\sqrt{Q}} (1- e^{j 2 \pi
\Delta}) \boldsymbol{\Lambda} _{p}^a \boldsymbol{C}_{N_t}
\boldsymbol{\Lambda} _{p}^b,
\end{equation}
where \setlength{\arraycolsep}{0.0em}
\begin{eqnarray*}
\boldsymbol{\Lambda}_{p}^a &{} = {} &
\mathrm{diag}\{a_0^{-(-p)_P/P}, \cdots, a_{N_t-1}^{-(-p)_P/P} \}, \\
\boldsymbol{\Lambda}_{p}^b &{} = {}& \mathrm{diag}\{ b_0^{(-p)_P/P},
\cdots, b_{N_t-1}^{(-p)_P/P} \},
\end{eqnarray*} \setlength{\arraycolsep}{5pt}\hspace*{-3pt}
and $\boldsymbol{C}_{N_t}$ is the $N_t \times N_t$ square matrix
with its element given by $[\boldsymbol{C}_{N_t}] _{\mu, \mu'} = c
^{(\mu, \mu')}$. From the definitions of $a_{\mu}$ and $b_{\mu'}$,
we have
\begin{multline}
a_{\mu} \ne 0 \ \& \ a_{\mu} \ne b_{\mu} \ \& \\
a_{\mu} \ne a_{\mu'}, \ b_{\mu} \ne b_{\mu'}, \ a_{\mu} \ne
b_{\mu'}, \ \forall \mu \ne \mu'.
\end{multline}
Moreover, we also have (\ref{equ34}) as shown at the bottom of the
page. \addtocounter{equation}{1} Then, with the assumption that
$N_t$ and $Q$ are not very large (for example $N_t \le 8, Q =16$),
the determinant of $\boldsymbol{C}_{N_t}$ can be obtained with its
definition as follows \setlength{\arraycolsep}{0.0em}
\begin{eqnarray}\label{app-01}
\hspace*{-8pt} \mathrm{det} \{\boldsymbol{C}_{N_t}\} &{} = {} &
\prod \limits_{\mu =1} ^{N_t-1} { \prod \limits_{\mu'=0}^{\mu-1}
\prod \limits_{\mu''=0}^{N_t-1} { \frac{(b_\mu-b_{\mu'}) (a_{\mu'} -
a_\mu) a_{\mu''} }{
(a_\mu-b_{\mu'}) (a_{\mu'}-b_\mu) (a_{\mu''}-b_{\mu''})} } } \nonumber \\
&{} \ne {}& 0,
\end{eqnarray} \setlength{\arraycolsep}{5pt}\hspace*{-3pt}
which shows that $ \mathrm{rank} \{ \boldsymbol{C}_{N_t} \} = N_t$.
From (\ref{theor26}), we immediately obtain that $\mathrm{rank} \{
\boldsymbol{\Lambda}_{p} \} = N_t$. With the special diagonal
structure of $\boldsymbol{\Lambda} $, we establish that $
\mathrm{rank} \{ \boldsymbol{\Lambda} \} = N_tP$ and then $
\mathrm{rank} \{ \boldsymbol{\mathcal{C}} \} = N_tP$. Exploiting the
following relationship
\begin{multline}
\mathrm{rank} \{ \boldsymbol{A}_0 \} + \mathrm{rank} \{
\boldsymbol{A}_1 \} - n \le \mathrm{rank} \{ \boldsymbol{A}_0
\boldsymbol{A}_1 \} \\
\le  \mathrm{min} \{ \mathrm{rank} \{
\boldsymbol{A}_0 \}, \mathrm{rank} \{ \boldsymbol{A}_1 \}\},
\end{multline}
where $\boldsymbol{A}_0$ has $n$ columns and $\boldsymbol{A}_1$ has
$n$ rows, we further establish that $\mathrm{rank} \{
\boldsymbol{\mathcal{C}} \boldsymbol{\breve{S}}
\boldsymbol{\breve{F}} \} = N_t L$ for \textsl{any} non-integer
$\Delta$ (recall that $P \ge L $ in our design condition). Hence,
for \textsl{any} $\boldsymbol{h} ^{(\nu,\mu)} (\ne
\boldsymbol{0}_L)$, we immediately obtain from its definition that
$\boldsymbol{d} _{\nu} (\Delta) \ne \boldsymbol{0}_{N-N_tP}$ and
then $\boldsymbol{d}(\Delta) \ne \boldsymbol{0}_{N_r(N-N_tP)}$.

    2) When $\Delta$ is an integer: By exploiting the structure of
$\boldsymbol{\mathcal{C}}$, $\boldsymbol{d}_{\nu}(\Delta)$ can be
transformed into the following equivalent form
\begin{multline}
\boldsymbol{d}_{\nu}(\Delta) = [(\boldsymbol{d}^{(\nu, 0)}
(\Delta))^T, (\boldsymbol{d}^{(\nu, 1)} (\Delta))^T, \\ \cdots,
(\boldsymbol{d}^{(\nu, \omega)} (\Delta))^T, \cdots,
(\boldsymbol{d}^{(\nu, Q-N_t -1)} (\Delta))^T]^T,
\end{multline}
where
\begin{multline*}
\boldsymbol{d}^{(\nu, \omega)} (\Delta) = \\
\sum \limits_{\mu=0} ^{N_t-1} \left\{
\boldsymbol{\mathcal{C}}^{(\omega, \mu)} \mathrm{diag}\{
\boldsymbol{\tilde{s}}_{\mu} \} \boldsymbol{\Theta} _{i_\mu}^T
\boldsymbol{F}_N [\boldsymbol{I}_L, \boldsymbol{0}_{L \times
(N-L)}]^T \boldsymbol{h}^{(\nu,\mu)} \right\}.
\end{multline*}
For \textsl{any} integer $\Delta \in (-Q, 0) \cup (0, Q)$, with our
design condition we have
\begin{equation}
(\boldsymbol{1}_Q - \boldsymbol{l})^T \boldsymbol{l} ^{(\Delta)} >
0.
\end{equation}
Without loss of generality, suppose $(i_{\mu} +\Delta)_Q =
z_{\omega}$. Then, together with condition (C0), which implies
$i_{\mu} = i_{\mu'}\ \mathrm{iff}. \ \mu = \mu'$, we have
\begin{equation}\label{app-42}
\boldsymbol{\mathcal{C}}^{(\omega, \mu)} = \sqrt{N} \boldsymbol{I}_P
\ \&  \ \boldsymbol{\mathcal{C}}^{(\omega, \mu')} =
\boldsymbol{0}_{P \times P}, \mathrm{if} \ \mu' \ne \mu.
\end{equation}
Hence,
\begin{multline}
\boldsymbol{d}^{(\nu, \omega)} (\Delta) = \\
\sqrt{N} \mathrm{diag}\{ \boldsymbol{\tilde{s}}_{\mu} \}
\boldsymbol{\Theta} _{i_\mu}^T \boldsymbol{F}_N [\boldsymbol{I}_L,
\boldsymbol{0}_{L \times (N-L)}]^T \boldsymbol{h}^{(\nu,\mu)}.
\end{multline}
With the condition $P \ge L$, we immediately obtain
$\boldsymbol{d}^{(\nu, \omega)}(\Delta) \ne \boldsymbol{0}_{P}$ and
then $\boldsymbol{d}(\Delta) \ne \boldsymbol{0}_{N_r(N-N_tP)}$ for
\textsl{any} $\boldsymbol{h}^{(\nu,\mu)}(\ne \boldsymbol{0}_L)$.

    Combining the above two cases, we draw the conclusion that $\boldsymbol{d}(\Delta)
\ne \boldsymbol{0}_{N_r(N-N_tP)}$ for \textsl{any} $\boldsymbol{h}
^{(\nu,\mu)}(\ne \boldsymbol{0}_L)$ and \textsl{any} $\Delta \in
(-Q, 0) \cup (0, Q)$. This completes the proof.

\begin{figure*}[!b]
\hrulefill
\normalsize
\setcounter{mytempeqncnt}{\value{equation}}
\setcounter{equation}{40}
\begin{multline}\label{cor46}
\mathfrak{P}(\Delta_i) = \sum \limits _{\nu = 0} ^{N_r-1} \sum
\limits _{\mu = 0} ^{N_t-1} \sum \limits _{\mu' = 0} ^{N_t-1}
\left\{ (\boldsymbol{h}^{(\nu, \mu)})^H \boldsymbol{M}
^{(\mu, \mu', \mu)} (\Delta_i) \boldsymbol{h}^{(\nu, \mu)} \right\} \\
+ \sum \limits _{\nu = 0} ^{N_r-1} \sum \limits _{\mu = 0} ^{N_t-1}
\sum \limits _{\mu' = 0} ^{N_t-1} \sum \limits _{\mu'' = 0, \mu''
\ne \mu} ^{N_t-1} \left\{ (\boldsymbol{h}^{(\nu, \mu)})^H
\boldsymbol{M} ^{(\mu, \mu', \mu'')} (\Delta_i)
\boldsymbol{h}^{(\nu, \mu'')} \right\},
\end{multline}
\setcounter{equation}{\value{mytempeqncnt}}
\end{figure*}

\begin{figure*}[!b]
\vspace*{-15pt}
\normalsize
\setcounter{mytempeqncnt}{\value{equation}}
\setcounter{equation}{43} \setlength{\arraycolsep}{0.0em}
\begin{eqnarray}\label{cor49}
\boldsymbol{M} ^{(\mu, \mu', \mu)} (\Delta_i) &{}={}&
\biggl\{\frac{P} {N_t} + \frac{2P}{N_tQ} \Re \left\{ {[(Q-1)e^{j
\theta} -Q e^{j 2
\theta} + e^{j (Q+1)\theta}]} / {(1-e^{j \theta})^2} \right\} \biggl\} \cdot \boldsymbol{I}_L \nonumber \\
&{}={}& \frac {P}{N_t Q} \cdot \frac{1-\cos (Q\theta)}{1-\cos
\theta} \cdot \boldsymbol{I}_L.
\end{eqnarray}\setlength{\arraycolsep}{5pt}\hspace*{-3pt}
\setcounter{equation}{\value{mytempeqncnt}}
\end{figure*}

\begin{figure*}[!b]
\vspace*{-15pt}
\normalsize
\setcounter{mytempeqncnt}{\value{equation}}
\setcounter{equation}{44} \setlength{\arraycolsep}{0.0em}
\begin{eqnarray}\label{cor50}
[\boldsymbol{M} ^{(\mu, \mu', \mu'')} (\Delta_i)]_{l,l'} &{}={}&
{\alpha (\mu, \mu''; l, l')} / {Q} \cdot (1-e^{j Q\theta}) ^2 / [
(1-e^{j \theta})(e^{j\psi} - e^{j \theta}) e^{j (Q-1)\theta} ] \nonumber \\
&{}={}& {\alpha(\mu, \mu''; l, l')} / {Q} \cdot \frac{ 1-\cos
(Q\theta) } {\cos (\psi/2) - \cos(\theta-\psi/2)}  .
\end{eqnarray}\setlength{\arraycolsep}{5pt}\hspace*{-3pt}
\setcounter{equation}{\value{mytempeqncnt}}
\end{figure*}

\vspace*{11pt}
\section{}
\vspace*{12pt}

    In this appendix, we will discuss the uniqueness of $\hat{\varepsilon} _i =
\varepsilon _i $ in maximizing the estimation metric for
\textsl{any} $\boldsymbol{h} ^{(\nu,\mu)}(\ne \boldsymbol{0}_L)$.
For the case that $\varepsilon = \varepsilon_i$, it follows
immediately from Theorem~\ref{theorem1} that the uniqueness of
$\hat{\varepsilon}_i = \varepsilon_i$ in maximizing the estimation
metric for \textsl{any} $\boldsymbol{h}^{(\nu,\mu)}(\ne
\boldsymbol{0}_L)$ is guaranteed. Therefore, in the following, we
only consider the case that $\varepsilon \ne \varepsilon_i$, i.e.,
$\varepsilon_f \ne 0$.

    Define $\Delta_i = \varepsilon - \tilde{\varepsilon}_i  $.
Then, the estimation metric $\mathfrak{P}(\Delta_i)$ can be written
into an equivalent form as shown in (\ref{cor46}) at the bottom of
the next page. In (\ref{cor46}), \addtocounter{equation}{1}
\setlength{\arraycolsep}{0.0em}
\begin{eqnarray*}
\boldsymbol{M} ^{(\mu, \mu', \mu'')} (\Delta_i) &{} = {} &
\boldsymbol{S}_{\mu}^H \boldsymbol{G}_{\mu'}(\Delta_i)
\boldsymbol{S}_{\mu''}, \\
\boldsymbol{G}_{\mu}(\Delta_i) &{} = {}& \boldsymbol{D}_N
(-\Delta_i) \boldsymbol{F}_N^H \boldsymbol{\Theta} _{i_\mu}
\boldsymbol{\Theta} _{i_\mu}^T \boldsymbol{F}_N \boldsymbol{D}_N
(\Delta_i), \\
\boldsymbol{S}_{\mu} &{} = {}& \sqrt{N} \boldsymbol{F}_N^H
\boldsymbol{\Theta} _{i_\mu} \mathrm{diag}\{
\boldsymbol{\tilde{s}}_{\mu} \} \boldsymbol{\Theta} _{i_\mu}^T
\boldsymbol{F}_N [\boldsymbol{I}_L, \boldsymbol{0}_{L \times
(N-L)}]^T.
\end{eqnarray*} \setlength{\arraycolsep}{5pt}\hspace*{-3pt}
From the definition of $\boldsymbol{G}_{\mu}(\Delta_i)$, we can
obtain 
\begin{multline}
\boldsymbol{G}_{\mu}(\Delta_i) = \\
\left[ {\begin{array}{*{20}l}
{\boldsymbol{G}_\mu^{(0 , 0 )}(\Delta_i) } &  \cdots  & {\boldsymbol{G}_\mu^{( 0 , Q-1 )}(\Delta_i) }  \\
\vdots  &  \boldsymbol{G}_\mu^{( q , q' )}(\Delta_i)  &  \vdots   \\
{\boldsymbol{G}_\mu^{( Q-1 , 0 )}(\Delta_i) } & \cdots & {\boldsymbol{G}_\mu^{( Q-1 , Q-1 )}(\Delta_i) }  \\
\end{array} } \right],
\end{multline}
where
\[
\boldsymbol{G}_\mu^{( q , q' )}(\Delta_i) = {Q}^{-1} e^{j 2 \pi
(q-q') ( i_{\mu} - \Delta_i )/Q } \cdot \boldsymbol{I}_P.
\]
Define
\[
\boldsymbol{s}_{\mu} = \frac{1}{\sqrt{Q}} \boldsymbol{F}_P^H
\boldsymbol{\tilde{s}}_{\mu} .
\]
Then, we can also obtain that $\boldsymbol{S}_\mu$ is an $N \times
L$ column-wise circulant matrix with $[\boldsymbol{D}_Q (i_{\mu} P)
\boldsymbol{1}_Q] \otimes [\boldsymbol{D}_P (i_{\mu})
\boldsymbol{s}_{\mu}] $ being its first column vector. Hence, we
have
\begin{multline}\label{cor48} %
\hspace*{-20pt} [\boldsymbol{M} ^{(\mu, \mu', \mu'')}
(\Delta_i)]_{l,l'} = {\alpha
(\mu, \mu''; l, l')} / {Q} \\
\times \biggl \{ \sum\limits _{q=0}^{Q-1} \sum\limits
_{q'=0}^{Q-q-1} \left\{ e^{j q' \psi} e^{j q \theta} + e^{j (q+q')
\psi} e^{-j q \theta} \right\} - \sum\limits _{q=0}^{Q-1} e^{j q
\psi} \biggl \},
\end{multline}
where \setlength{\arraycolsep}{0.0em}
\begin{eqnarray*}
\alpha (\mu, \mu''; l, l') &{} = {} & e^{j 2 \pi (l i_{\mu} - l'
i_{\mu''}) / N}\cdot (\boldsymbol{s}_{\mu}^{(l)})^H \boldsymbol{D}_P
(i_{\mu''} - i_{\mu}) \boldsymbol{s}_{\mu''}^{(l')} ,\\
\theta &{} = {}& 2\pi (i_{\mu'} - i_{\mu} - \Delta_i ) /Q, \\
\psi &{} = {}& 2\pi (i_{\mu''} - i_{\mu}) /Q.
\end{eqnarray*} \setlength{\arraycolsep}{5pt}\hspace*{-3pt}

    When ${\mu} = {\mu''}$, imposing the condition that the
elements of $\boldsymbol{\tilde{s}} _{\mu}$ have constant amplitude,
which translates to the time-domain zero auto-correlation (ZAC)
property of $\boldsymbol{s}_{\mu}$ \cite{Minn_Periodic}, from
(\ref{cor48}) we immediately obtain (\ref{cor49}) as shown at the
bottom of the page. \addtocounter{equation}{1} When ${\mu} \ne
{\mu''}$, from (\ref{cor48}) we obtain (\ref{cor50}) as shown at the
bottom of the page. Define \addtocounter{equation}{1}
\[
\zeta ^{(\mu, \mu', \mu'')} (\Delta_i) = \left| {[\boldsymbol{M}
^{(\mu, \mu', \mu)} (\Delta_i)] _{l,l}} / {[\boldsymbol{M} ^{(\mu,
\mu', \mu'')} (\Delta_i)]_{l,l'}} \right|.
\]
Then, we immediately establish
\begin{multline}
\zeta ^{(\mu, \mu', \mu'')} (\Delta_i) = {P}/ [N_t \cdot
|\alpha(\mu, \mu''; l, l')| \hspace*{1pt}] \\
\times |\sin (\psi/2)| \cdot | \cot (\psi/2) - \cot(\theta/2) |.
\end{multline}

    The following remarks are now in order.

    \textit{Remark 1:} Both $[\boldsymbol{M} ^{(\mu, \mu', \mu)}
(\Delta_i)] _{l,l}$ and $[\boldsymbol{M} ^{(\mu, \mu', \mu'')}
(\Delta_i)]_{l,l'}$ are the period-$Q$ functions with respect to
$\Delta_i$. Fix the true CFO $\varepsilon$ and vary the candidate
ICFO $\tilde{\varepsilon}_i$ between $(-\lfloor Q/2 \rfloor,
Q-\lfloor Q/2 \rfloor]$. Then, we have
\begin{equation}
\Delta_i \in [\varepsilon -Q + \lfloor Q/2 \rfloor, \varepsilon +
\lfloor Q/2 \rfloor) \subset (-Q, Q).
\end{equation}
For description convenience, we employ $r(\Delta_i)$ to normalize
$\Delta_i$ from $[\varepsilon -Q + \lfloor Q/2 \rfloor, \varepsilon
+ \lfloor Q/2 \rfloor)$ to $[-\lfloor Q/2 \rfloor, Q-\lfloor Q/2
\rfloor)$, where
\begin{multline}
r(\Delta_i) = \Delta_i - [\mathrm{sign}(\Delta_i - Q + \lfloor
Q/2\rfloor) \\
+ \mathrm{sign}(\Delta_i + \lfloor Q/2 \rfloor)] Q/2 .
\end{multline}
Therefore, $[\boldsymbol{M} ^{(\mu, \mu', \mu)} (\Delta_i)] _{l,l}$
achieves its minimum $0$ when $r(\Delta_i) \in \{-\lfloor Q/2
\rfloor, -\lfloor Q/2 \rfloor +1, \cdots, $ $Q-\lfloor Q/2 \rfloor
-1 \} \setminus \{ r(i_{\mu'} - i_{\mu}) \}$ and its maximum $N/N_t$
when $r(\Delta_i) = r(i_{\mu'} - i_{\mu})$. While $\bigl|
[\boldsymbol{M} ^{(\mu, \mu', \mu'')} (\Delta_i)]_{l,l'} \bigl|$
achieves its minimum 0 when $r(\Delta_i) \in \{-\lfloor Q/2 \rfloor,
-\lfloor Q/2 \rfloor +1, \cdots, Q-\lfloor Q/2 \rfloor -1 \}$.

    \textit{Remark 2:} Both $[\boldsymbol{M} ^{(\mu, \mu', \mu)} (\Delta_i)] _{l,l}$
and $[\boldsymbol{M} ^{(\mu, \mu', \mu'')} (\Delta_i)]_{l,l'}$ are
the continuous functions with respect to $\Delta_i$. For any $\gamma
> 0$ (for example $\gamma = 10^{-2}$), there exists $\delta
>0$ (for example $\delta = 0.1$) which makes the relationships as shown in
(\ref{equ49}) and (\ref{equ50}) at the bottom of the next page hold.
\addtocounter{equation}{2}

\begin{figure*}[!b]
\vspace*{-9pt}
\hrulefill
\normalsize
\setcounter{mytempeqncnt}{\value{equation}}
\setcounter{equation}{48}
\begin{multline}\label{equ49}
[\boldsymbol{M} ^{(\mu, \mu', \mu)} (\Delta_i)] _{l,l}  < \gamma, \
\mathrm{if} \ r(\Delta_i) \in \bigcup \limits _{q = - \lfloor Q/2
\rfloor} ^{Q-\lfloor Q/2 \rfloor} \bigl\{ (q - \delta, q) \cup (q, q
+ \delta) \bigl\} \setminus  \bigl \{ (r(i_{\mu'} - i_{\mu}) -
\delta, r(i_{\mu'} - i_{\mu}) + \delta ) \bigl\} \\
\setminus \bigl \{ (- \lfloor Q/2 \rfloor - \delta, - \lfloor Q/2
\rfloor) \cup (Q-\lfloor Q/2 \rfloor, Q-\lfloor Q/2 \rfloor +
\delta) \bigl \},
\end{multline} 
\setcounter{equation}{\value{mytempeqncnt}}
\end{figure*}

\begin{figure*}[!b]
\vspace*{-15pt}
\normalsize
\setcounter{mytempeqncnt}{\value{equation}}
\setcounter{equation}{49}
\begin{multline}\label{equ50}
\bigl| [\boldsymbol{M} ^{(\mu, \mu', \mu'')} (\Delta_i)]_{l,l'}
\bigl|  < \gamma, \ \mathrm{if} \ r(\Delta_i) \in \bigcup \limits
_{q = - \lfloor Q/2 \rfloor} ^{Q-\lfloor Q/2 \rfloor} \bigl\{ (q -
\delta, q) \cup (q, q + \delta) \bigl\} \\
\setminus \bigl \{ (- \lfloor Q/2 \rfloor - \delta, - \lfloor Q/2
\rfloor) \cup (Q-\lfloor Q/2 \rfloor, Q-\lfloor Q/2 \rfloor +
\delta) \bigl \}.
\end{multline}
\setcounter{equation}{\value{mytempeqncnt}}
\end{figure*}

    \textit{Remark 3:} Assume $0 < |\varepsilon_f| < 0.5 $.
Then, for $r(\varepsilon_i - \tilde{\varepsilon}_i) \in \{-\lfloor
Q/2 \rfloor, -\lfloor Q/2 \rfloor +1, \cdots, Q-\lfloor Q/2 \rfloor
-1 \}$, we have (\ref{equ51}) as shown at the bottom of the next
page. \addtocounter{equation}{1}

\begin{figure*}[!b]
\vspace*{-15pt}
\normalsize
\setcounter{mytempeqncnt}{\value{equation}}
\setcounter{equation}{50}
\begin{multline}\label{equ51}
{{[\boldsymbol{M} ^{(\mu, \mu', \mu)} (\Delta_i | r(\varepsilon_i -
\tilde{\varepsilon}_i) = r(i_{\mu'} - i_{\mu}) )] _{l,l} }
\mathord{\left/
 {\vphantom {{[\boldsymbol{M} ^{(\mu, \mu', \mu)} (\Delta_i |
r(\varepsilon_i - \tilde{\varepsilon}_i) = r(i_{\mu'} - i_{\mu}))]
_{l,l}} {[\boldsymbol{M} ^{(\mu, \mu', \mu)} (\Delta_i |
r(\varepsilon_i - \tilde{\varepsilon}_i) \ne r(i_{\mu'} - i_{\mu})}
)] _{l,l} }} \right.
 \kern-\nulldelimiterspace} {[\boldsymbol{M} ^{(\mu, \mu', \mu)} (\Delta_i |
r(\varepsilon_i - \tilde{\varepsilon}_i) \ne r(i_{\mu'} - i_{\mu})}
)] _{l,l} } \\
= {\sin ^2 \{\pi [\varepsilon_f + r(\varepsilon_i -
\tilde{\varepsilon}_i) - r(i_{\mu'} - i_{\mu}) ]/Q\}} / { \sin^2
(\pi\varepsilon_f /Q) } \bigl |_{r(\varepsilon_i -
\tilde{\varepsilon}_i) \ne r(i_{\mu'} - i_{\mu})} \,  > 1 .
\end{multline}
\setcounter{equation}{\value{mytempeqncnt}}
\end{figure*}

    \textit{Remark 4:} $\zeta ^{(\mu, \mu', \mu'')} (\Delta_i)$ is the period-$Q$
continuous function with respect to $\Delta_i$, which achieves its
minimum 0 when $r(\Delta_i) = r(i_{\mu'} - i_{\mu''})$ and its
maximum $+ \infty$ when $r(\Delta_i) = r(i_{\mu'} - i_{\mu})$.
Impose the condition that $|\alpha (\mu, \mu''; l, l')| < 1/N_t$ for
$\mu \ne \mu''$. Then, there exists $\chi > 0$ (for example $\chi =
5$) which makes the relationship as shown in (\ref{equ52}) at the
bottom of the next page hold. \addtocounter{equation}{1}

\begin{figure*}[!b]
\vspace*{-15pt}
\normalsize
\setcounter{mytempeqncnt}{\value{equation}}
\setcounter{equation}{51}
\begin{multline}\label{equ52}
\zeta ^{(\mu, \mu', \mu'')} ( \Delta_i ) > P \cdot |\sin (\psi/2)|
\cdot | \cot (\psi/2) - \cot(\psi/2 + \pi \delta /Q ) | > \chi, \\
\mathrm{if} \ r(\Delta_i) \in \bigl\{ (- \lfloor Q/2 \rfloor,
r(i_{\mu'} - i_{\mu''}) - \delta) \cup  (r(i_{\mu'} - i_{\mu''}) +
\delta , Q-\lfloor Q/2 \rfloor) \bigl\}.
\end{multline}
\setcounter{equation}{\value{mytempeqncnt}}
\end{figure*}

\begin{figure*}[!b]
\vspace*{-15pt}
\normalsize
\setcounter{mytempeqncnt}{\value{equation}}
\setcounter{equation}{52}
\begin{multline}\label{P40}
\mathfrak{P}(\Delta_i) \doteq  0, \ \mathrm{if} \ r(\Delta_i) \in
\bigcup \limits _{q = - \lfloor Q/2 \rfloor} ^{Q-\lfloor Q/2
\rfloor} \bigl\{ (q - \delta, q) \cup (q, q + \delta) \bigl\}
\setminus \bigl\{ \bigcup \limits _{\mu, \mu' = 0} ^{N_t-1}
\bigl\{(r(i_{\mu'} - i_{\mu}) - \delta,
r(i_{\mu'} - i_{\mu}) + \delta ) \bigl \} \bigl\} \\
 \setminus \bigl \{ (- \lfloor Q/2 \rfloor - \delta,
- \lfloor Q/2 \rfloor) \cup (Q-\lfloor Q/2 \rfloor, Q-\lfloor Q/2
\rfloor + \delta)  \bigl \},
\end{multline}
\setcounter{equation}{\value{mytempeqncnt}}
\end{figure*}

\begin{figure*}[!b]
\vspace*{-15pt}
\normalsize
\setcounter{mytempeqncnt}{\value{equation}}
\setcounter{equation}{53}
\begin{multline}\label{P41}
\mathfrak{P}(\Delta_i) \doteq \sum \limits _{\nu = 0} ^{N_r-1} \sum
\limits _{\mu = 0} ^{N_t-1} \sum \limits _{\mu' = 0} ^{N_t-1}
\left\{ (\boldsymbol{h}^{(\nu, \mu)})^H \boldsymbol{M} ^{(\mu,
\mu', \mu)} (\Delta_i) \boldsymbol{h}^{(\nu, \mu)} \right\}, \\
\mathrm{if} \ r(\Delta_i) \in \bigcup \limits _{q = - \lfloor Q/2
\rfloor} ^{Q-\lfloor Q/2 \rfloor -1 } \bigl\{ [q + \delta, q+1
-\delta ] \bigl\} \bigcup \bigl\{ \bigcup \limits _{\mu, \mu' = 0}
^{N_t-1} \bigl\{ (r(i_{\mu'} - i_{\mu}) - \delta, r(i_{\mu'} -
i_{\mu}) + \delta) \bigl \} \bigl\}.
\end{multline}
\setcounter{equation}{\value{mytempeqncnt}}
\end{figure*}

    We use an example as shown in Fig. \ref{Jiang_Paper-TW-Oct-06-0846_fig5}
to illustrate the above remarks.

    Assume
\[
\mathrm{E}\{[\boldsymbol{h}^{(\nu,\mu)}]_l^* [\boldsymbol{h}^{(\nu
',\mu ')}]_{l'} \} = 0, \ \forall (\nu,\mu) \ne (\nu ',\mu ').
\]
Then, it follows from Remark 2 and 4 that (\ref{P40}) and
(\ref{P41}) as shown at the bottom of the page can be obtained,
respectively. \addtocounter{equation}{2} According to its
definition, $[\boldsymbol{M} ^{(\mu, \mu', \mu)} (\Delta_i)]_{l,l}$
must be one of the $Q$ possible values \textsl{whatever} $i_{\mu}$
or $i_{\mu'}$ is when we vary $\tilde{\varepsilon}_i$ from
$(-\lfloor Q/2 \rfloor + 1)$ to $(Q-\lfloor Q/2 \rfloor)$. From
Remark 3, we immediately obtain that
\begin{multline}
{{[\boldsymbol{M} ^{(\mu, \mu', \mu)} (\varepsilon_f)] _{l,l}}
\mathord{\left/ {\vphantom {{[\boldsymbol{M} ^{(\mu, \mu', \mu)}
(\varepsilon_f)] _{l,l}} {[\boldsymbol{M} ^{(\mu, \mu', \mu)}
(\varepsilon_f + ( i_{\mu'} - i_{\mu} )_Q)] _{l,l}}}} \right.
 \kern-\nulldelimiterspace} {[\boldsymbol{M} ^{(\mu, \mu', \mu)}
(\varepsilon_f + ( i_{\mu'} - i_{\mu} )_Q)] _{l,l}}} \gg 1, \\
\mathrm{if} \ ( i_{\mu'} - i_{\mu} )_Q >1,
\end{multline}
and the items involved in the right hand side of (\ref{P41}) achieve
their maximum when $r(\varepsilon_i - \tilde{\varepsilon}_i) =
r(i_{\mu'} - i_{\mu})$ for \textsl{any}
$\boldsymbol{h}^{(\nu,\mu)}(\ne \boldsymbol{0}_L)$.
\begin{figure}[!t]\vspace*{-2pt}
\centering
\includegraphics[width=0.48\textwidth]{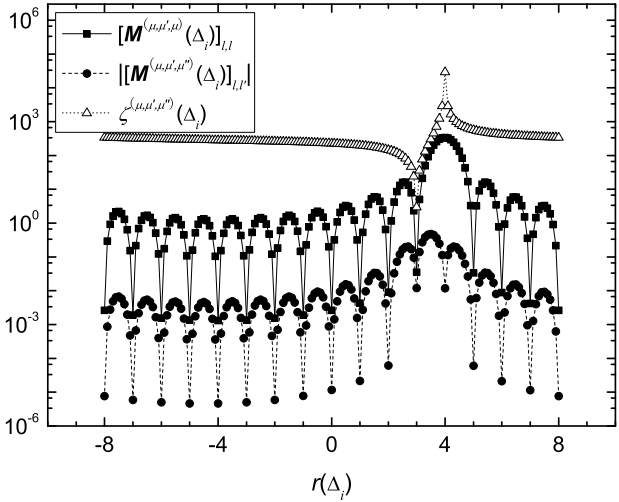}
\captionstyle{mycaptionstyle} \caption{ $[\boldsymbol{M} ^{(\mu,
\mu', \mu)} (\Delta_i)] _{l,l}$, $\bigl |[\boldsymbol{M} ^{(\mu,
\mu', \mu'')} (\Delta_i)]_{l,l'}\bigl|$ and $\zeta ^{(\mu, \mu',
\mu'')} ( \Delta_i )$ versus $r(\Delta_i)$ with $P=64$, $Q=16$, $N_t
= 3$, $i_{\mu} = 3$, $i_{\mu '} = 7$, $i_{\mu ''} = 4$. }
\label{Jiang_Paper-TW-Oct-06-0846_fig5}
\end{figure}
Then, we can establish that there are $N_r N_t$ items involved in
the right hand side of (\ref{P41}) which achieve the maximum when $
\tilde{\varepsilon}_i = \varepsilon_i $ (i.e., $\mu = \mu'$), and at
most $N_r \bigl(N_t-\mathop {\min}\limits_{1 \le q \le Q-1}\left\{
(\boldsymbol{1}_Q - \boldsymbol{l})^T \boldsymbol{l}^{(q)} \right\}
\bigl)$ items that achieve the maximum when $ \tilde{\varepsilon}_i
\ne \varepsilon_i $ (i.e., $\mu \ne \mu'$). Due to the random
snapshot channel energies, a closed form of training design
conditions to yield the uniqueness of $\hat{\varepsilon}_i =
{\varepsilon}_i$ is intractable. However, it follows from condition
(C3) and the above analysis that the reliability of the uniqueness
can be maximized by designing the training sequences to maximize
$\mathop {\min}\limits_{ \mu' \ne \mu} \bigl\{( i_{\mu'} - i_{\mu}
)_Q \bigl\} $ and $\mathop {\min}\limits_{1 \le q \le Q-1}\left\{
(\boldsymbol{1}_Q - \boldsymbol{l})^T \boldsymbol{l}^{(q)}
\right\}$.

\section*{Acknowledgment}
    The authors would like to thank the anonymous reviewers for
their valuable comments which helped to improve the quality of the
paper greatly.

\normalem

\bibliographystyle{IEEEtran}
\bibliography{Jiang_Paper-TW-Oct-06-0846}

\vspace*{-15pt} \vspace*{-2\baselineskip}
\begin{biography}[{\includegraphics[width=1in,height
=1.25in,clip,keepaspectratio]{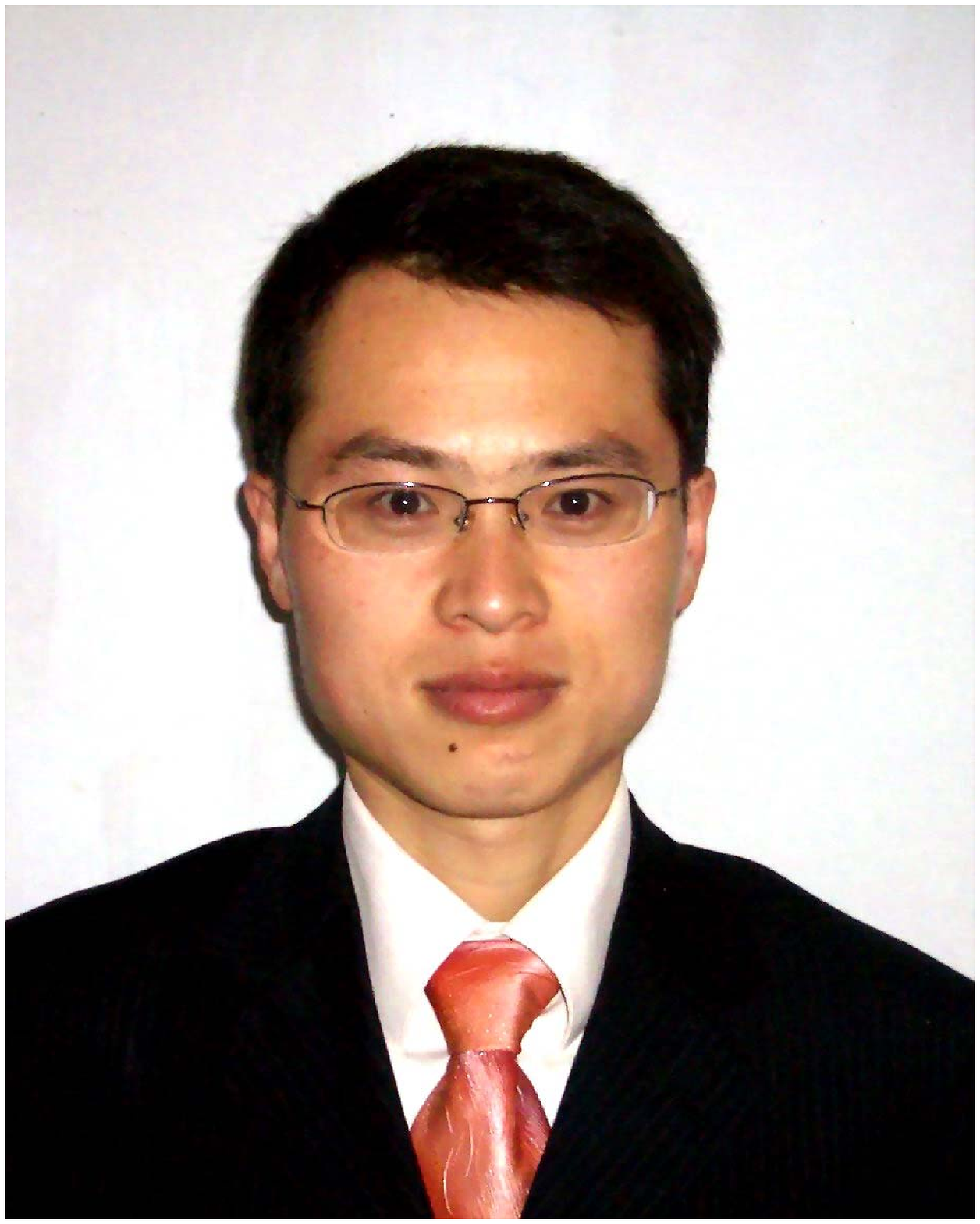}}]{Yanxiang
Jiang (S'03)} received the B.S. degree in electrical engineering
from Nanjing University, Nanjing, China, in 1999, and the M.E.
degree in radio engineering from Southeast University, Nanjing,
China, in 2003. He is currently working toward the Ph.D. degree at
the National Mobile Communications Research Laboratory, Southeast
University, Nanjing, China.

    He received the NJU excellent graduate honor from Nanjing University
in 1999. His research interests include mobile communications,
wireless signal processing, parameter estimation, synchronization,
signal design, and digital implementation of communication systems.
\end{biography}

\vspace*{-2\baselineskip}

\begin{biography}[{\includegraphics[width=1in,height
=1.25in,clip,keepaspectratio]{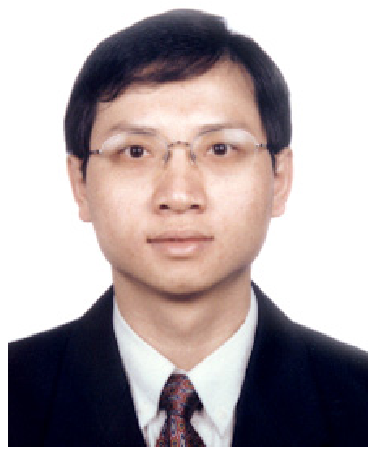}}]{Hlaing
Minn (S'99-M'01)} received his B.E. degree in Electronics from
Yangon Institute of Technology, Yangon, Myanmar, in 1995, M.Eng.
degree in Telecommunications from Asian Institute of Technology
(AIT), Pathumthani, Thailand, in 1997 and Ph.D. degree in Electrical
Engineering from the University of Victoria, Victoria, BC, Canada,
in 2001.

    He was with the Telecommunications Program in AIT as a laboratory
supervisor during 1998. He was a research assistant from 1999 to
2001 and a post-doctoral research fellow during 2002 in the
Department of Electrical and Computer Engineering at the University
of Victoria. Since September 2002, he has been with the Erik Jonsson
School of Engineering and Computer Science, the University of Texas
at Dallas, USA, as an Assistant Professor. His research interests
include wireless communications, statistical signal processing,
error control, detection, estimation, synchronization, signal
design, and cross-layer design. He is an Editor for the \emph{IEEE
Transactions on Communications}.
\end{biography}

\vspace*{-2\baselineskip}

\begin{biography}[{\includegraphics[width=1in,height
=1.25in,clip,keepaspectratio]{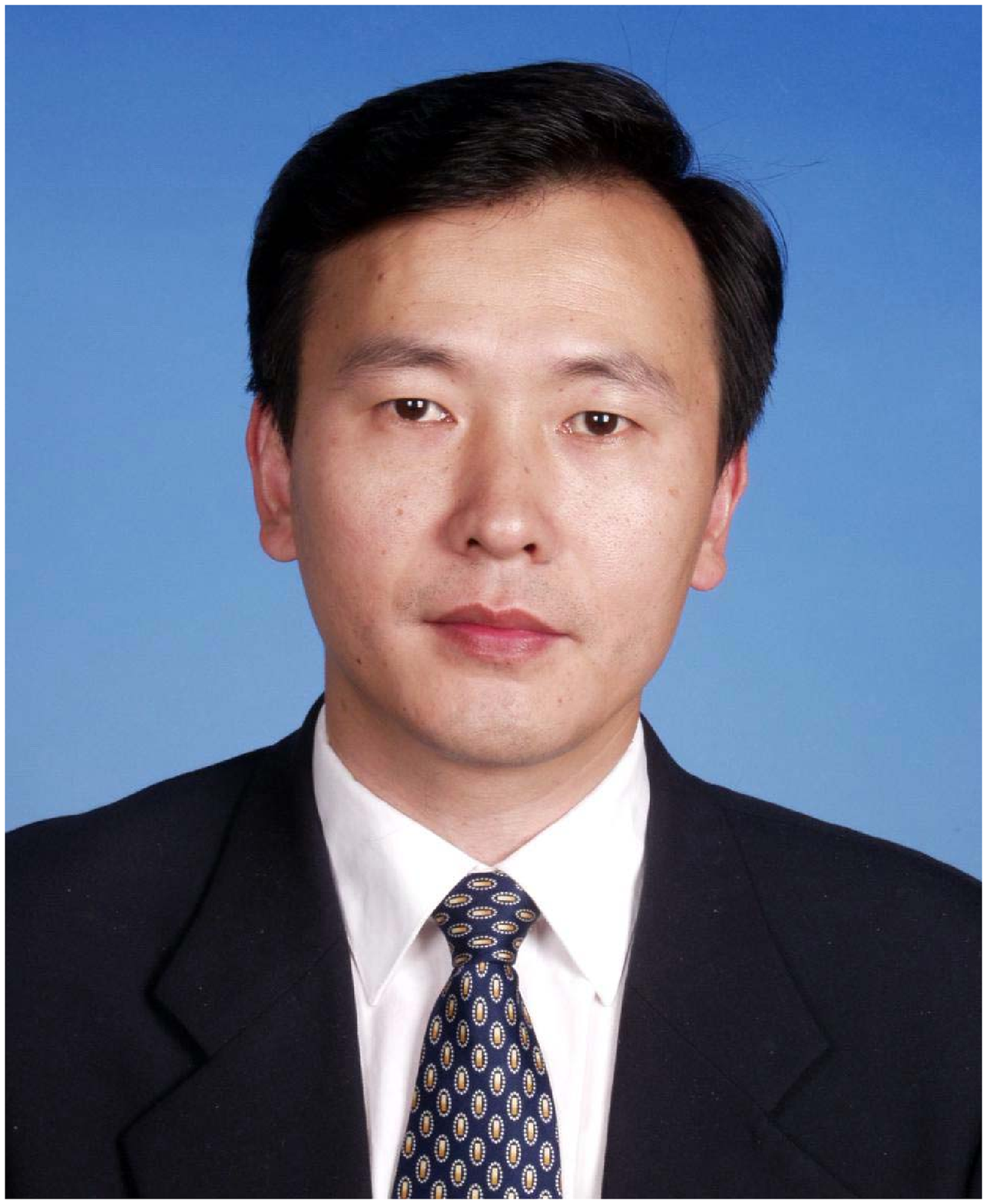}}]{
Xiqi Gao (SM'07)} received the Ph.D. degree in electrical
engineering from Southeast University, Nanjing, China, in 1997. He
joined the Department of Radio Engineering, Southeast University, in
April 1992. Now he is a professor of information systems and
communications. From September 1999 to August 2000, he was a
visiting scholar at Massachusetts Institute of Technology,
Cambridge, and Boston University, Boston, MA. His current research
interests include broadband multi-carrier transmission for beyond 3G
mobile communications, space-time wireless communications, iterative
detection/decoding, signal processing for wireless communications.

  Dr. Gao received the Science and Technology Progress Awards of
the State Education Ministry of China in 1998 and 2006. He is
currently serving as an editor for the \emph{IEEE Transactions on
Wireless Communications}.

\end{biography}

\vspace*{-2\baselineskip}

\begin{biography}[{\includegraphics[width=1in,height
=1.25in,clip,keepaspectratio]{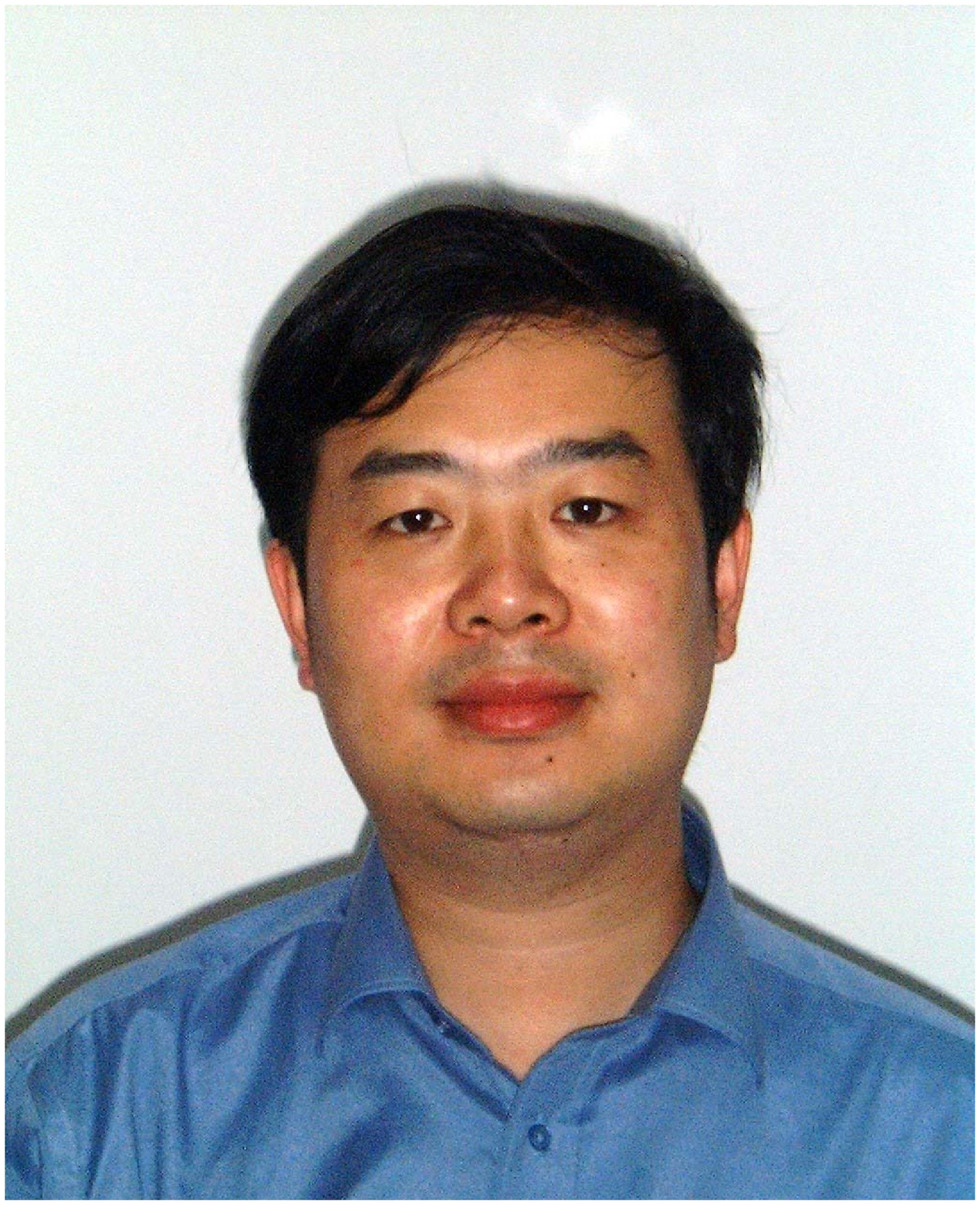}}]{Xiaohu
You} received the M.S. and Ph.D. degrees in electrical engineering
from Southeast University, Nanjing, China, in 1985 and 1988,
respectively.

    Since 1990 he has been working with National Mobile Communications
Research Laboratory at Southeast University, where he holds the
ranks of professor and director. From 1993 to 1997 he was engaged,
as a team leader, in the development of China's first GSM and CDMA
trial systems. He was the Premier Foundation Investigator of the
China National Science Foundation in 1998. From 1999 to 2001 he was
on leave from Southeast University, working as the chief director of
China's 3G (C3G) Mobile Communications R$\&$D Project. He is
currently responsible for organizing China's B3G R$\&$D activities
under the umbrella of the National 863 High-Tech Program, and he is
also the chairman of the China 863-FuTURE Expert Committee. He has
published two books and over 20 IEEE journal papers in related
areas.  His research interests include mobile communications,
advanced signal processing, and applications.
\end{biography}

\vspace*{-2\baselineskip}

\begin{biography}[{\includegraphics[width=1in,height
=1.25in,clip,keepaspectratio]{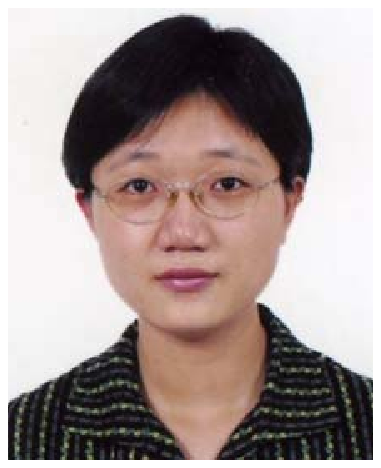}}]{Yinghui
Li (S'05)} received the B.E. and M.S. degree in Electrical
Engineering from the Nanjing University of Aeronautics and
Astronautics, Nanjing, China, in 2000 and 2003, respectively. She is
currently working toward the Ph.D. degree in Electrical Engineering
at University of Texas at Dallas. Her research interests are in the
applications of statistical signal processing in synchronization,
channel estimation and detection problems in broadband wireless
communications.
\end{biography}

\end{document}